\documentclass[11pt]{article}
\usepackage[utf8]{inputenc}
\input{epsf}
\usepackage[utf8]{inputenc}
\usepackage{amsmath}
\usepackage{amssymb}
\input{epsf}
\usepackage{natbib}
\usepackage[margin=1in]{geometry}
\usepackage[raggedright]{titlesec}
\usepackage[hang,flushmargin]{footmisc} 
\usepackage{graphicx}
\usepackage{framed}
\usepackage{color}
\usepackage{yfonts}
\usepackage[dvipsnames]{xcolor}
\usepackage{pifont}
\usepackage{amsmath,amssymb}
\usepackage{multirow}
\usepackage{amsfonts}
\usepackage{subfigure}
\usepackage{caption}
\usepackage{threeparttable}
\usepackage{esvect}
\usepackage{pgf}
\usepackage{tikz}
\usepackage{tikz-cd}
\usepackage{fancyvrb}
\usetikzlibrary{positioning}
\usepackage{bbm}
\usepackage{colortbl}
\usepackage{pdflscape}
\usepackage{booktabs}
\usepackage{amsmath}
\usepackage{mathdots}
\usepackage{yhmath}
\usepackage{cancel}
\usepackage{siunitx}
\usepackage{array}
\usepackage{multirow}
\usepackage{gensymb}
\usepackage{tabularx}
\usepackage{booktabs}
\usetikzlibrary{fadings}
\usetikzlibrary{patterns}
\usetikzlibrary{shadows.blur}
\usetikzlibrary{shapes}
\usepackage{placeins}
\usepackage{comment}

\usepackage{hyperref}
\hypersetup{
    colorlinks=true,
    linkcolor=NavyBlue,
    filecolor=magenta,      
    urlcolor=cyan,
    citecolor=black
}
 
\usepackage{rotating}
\usepackage{ntheorem}
\makeatletter
\renewtheoremstyle{plain} % Adds automatic line break, if heading is too long
{\item{\theorem@headerfont ##1\ ##2\theorem@separator}~}
{\item{\theorem@headerfont ##1\ ##2\ (##3)\theorem@separator}~}
\makeatother

\theoremheaderfont{\normalfont\bfseries}

\newtheorem{theorem}{Theorem}[section]

\newtheorem{lemma}{Lemma}[section]

\newtheorem{corollary}{Corollary}[section]
\newtheorem{example}{Example}[section]
\newtheorem{assumption}{Assumption}

\newenvironment{proof}{\paragraph{Proof:}}{\hfill$\square$}

\newcommand{\iid}{\stackrel{iid}{\sim}}
\newcommand{\cip}{\stackrel{p}{\rightarrow}}
\newcommand{\cid}{\stackrel{d}{\rightarrow}}
\newcommand{\E}{\mathbb{E}}

\newcommand{\var}{\text{var}}

\newcommand{\cov}{\text{cov}}
\newcommand{\cor}{\text{cor}}

\newcommand{\R}{\mathbbm{R}}

\newcommand{\1}{\mathbbm{1}}

\newcommand{\indep}{\raisebox{0.05em}{\rotatebox[origin=c]{90}{$\models$}}}%Perpendicular symbol
\definecolor{shadecolor}{gray}{0.9}

\tikzset{every picture/.style={line width=0.75pt}} %set default line width to 0.75pt        \definecolor{lgray}{rgb}{0.9, 0.9, 0.9}
\newcolumntype{a}{>{\columncolor{lgray}c}}

\usepackage{enumitem}
\newlist{Step}{enumerate}{2}
\setlist[Step]{label={{Step \arabic*.}}, leftmargin=*}

\usepackage{setspace}
 \makeatletter
\newcommand\circled[1]{%
  \mathpalette\@circled{#1}%
}
\newcommand\@circled[2]{%
  \tikz[baseline=(math.base)] \node[draw,circle,inner sep=2pt] (math) {$\m@th#1#2$};%
}
\makeatother

 \makeatletter
\newcommand\circledblue[1]{%
  \mathpalette\@circledblue{#1}%
}
\newcommand\@circledblue[2]{%
  \tikz[baseline=(math.base)] \node[draw,circle, fill=blue!20, inner sep=2pt] (math) {$\m@th#1#2$};%
 }

\renewenvironment{abstract}
 {\begin{center}\normalsize\textsc{Abstract}%
 \end{center}\begin{quote}\normalsize}
 {\end{quote}}

\title{Design Sensitivity and Its Implications for Weighted Observational Studies\thanks{The authors would like to thank Licong Lin, Erin Hartman, Avi Feller, and Kevin Guo for their helpful comments and feedback. Melody Huang is supported by the National Science Foundation Graduate Research Fellowship under Grant No. 2146752, Dan Soriano is supported by the National Science Foundation under Grant No. 1745640, and Samuel D. Pimentel is supported by the National Science Foundation under Grant No. 2142146. Any opinion, findings, and conclusions or recommendations expressed in this material are those of the authors(s) and do not necessarily reflect the views of the National Science Foundation.}}
\author{Melody Huang\thanks{Harvard University, USA (\texttt{melodyhuang@fas.harvard.edu})}, Dan Soriano\thanks{University of California, Berkeley, USA (\texttt{dan\_soriano@berkeley.edu})}, and Samuel D. Pimentel\thanks{University of California, Berkeley, USA (\texttt{spi@berkeley.edu})}}
\date{\today}

\begin{document}

\maketitle 

\begin{abstract}
Sensitivity to unmeasured confounding is not typically a primary consideration in designing treated-control comparisons in observational studies. We introduce a framework for researchers to optimize robustness to omitted variable bias at the design stage using a measure called design sensitivity. Design sensitivity describes the asymptotic power of a sensitivity analysis, and transparently assesses the impact of different estimation strategies on sensitivity.  We apply this general framework to two commonly-used sensitivity models, the marginal sensitivity model and the variance-based sensitivity model. By comparing design sensitivities, we interrogate how key features of weighted designs, including choices about trimming of weights and model augmentation, impact robustness to unmeasured confounding, and how these impacts may differ for the two different sensitivity models. We illustrate the proposed framework on a study examining drivers of support for the 2016 Colombian peace agreement.

\end{abstract}

\clearpage

\begingroup
\allowdisplaybreaks
\doublespacing 
\section{Introduction}
Increasingly, observational studies are being used to answer causal questions in the social and biomedical sciences. Estimating causal effects in observational settings often requires an assumption that unmeasured confounding is absent. In practice, this assumption is not testable and often untenable. Recent literature has introduced sensitivity analyses to assess the potential impact of an unobserved confounder on a study’s results \citep{zhao2019sensitivity, soriano2021interpretable, jin2022sensitivity, ishikawa2023kernel}. However, sensitivity analysis remains underutilized in practice, and sensitivity to unmeasured confounding is not typically a primary consideration in designing the treated-control comparison.

Given that findings from observational studies cannot be viewed as reliable unless a sensitivity analysis demonstrates some robustness to unmeasured confounding, it is natural to approach the design of an observational study with unmeasured confounding in mind.  In the following paper, we introduce a measure called \textit{design sensitivity} for weighted observational studies. Design sensitivity describes the asymptotic power of a sensitivity analysis and can be computed prior to carrying out a study.  Computing and comparing design sensitivities across  possible weighted designs allows researchers to optimize for robustness to unmeasured confounding rather than treating sensitivity analysis as a post hoc secondary analysis.

The paper provides several key methodological contributions. First, while design sensitivity has been studied for matched observational studies  \citep[e.g.,][]{rosenbaum2004design, rosenbaum2015bahadur, howard2021uniform}, the notion of design sensitivity has not yet been introduced to the weighting literature. By introducing design sensitivity for weighted estimators, we interrogate how key features of weighted designs impact robustness to unmeasured confounding. In particular, we examine two design choices unique to weighted estimators---augmentation and trimming---and formalize how they impact the robustness of weighted estimators to unmeasured confounding. 

Second, we derive the existence of design sensitivity for a \textit{broad class of sensitivity models} for weighted estimators. This means that under our proposed framework, researchers can construct the design sensitivity for a wide variety of sensitivity models with only mild regularity conditions required. The generality of the framework allows researchers to flexibly estimate the design sensitivity for the sensitivity models of their choice. This is especially advantageous, given the variety of different approaches that have been proposed for performing sensitivity analysis for weighted estimators. We illustrate the framework by deriving closed form representations for the design sensitivities of two common sensitivity models: the marginal sensitivity model and the variance-based sensitivity model. These closed form representations allow researchers to mathematically examine how design choices affect different sensitivity models. 

The paper is structured as follows. In Section 2, we introduce the notation and background. We formalize the notion of a general sensitivity model for weighted estimators, and review existing sensitivity approaches. In Section 3, we introduce the idea of power for a sensitivity model and design sensitivity. In Section 4, we compare the design sensitivity for different design choices. Section 5 illustrates the proposed framework on an empirical analysis examining the drivers of support for the FARC peace agreement. Section 6 concludes. 

\section{Background}

\subsection{Set-Up, Notation, and Assumptions}
We consider an observational study of $n$ units sampled identically and independently from an infinite population. Define $Z$ as a treatment indicator, where $Z = 1$ if a unit is in the treatment group, and 0 otherwise. Furthermore, define $Y(1)$ and $Y(0)$ as the potential outcomes under treatment and control, respectively. Throughout, we will make the  stable unit treatment value assumption (SUTVA)--i.e., no interference or spillovers, such that the observed outcomes $Y_i$ can be written as $Y = Y(1) \cdot Z + Y(0) \cdot (1-Z)$ \citep{rubin1980SUTVA}. 

In randomized trials, researchers determine the probability of assignment to treatment or control for each individual, but in observational studies, propensities for treatment may co-vary with unobserved potential outcomes.  To permit unbiased estimation of treatment effects, researchers must measure all background covariates describing common variation in treatment and potential outcomes, as formalized in the following assumption.

\begin{assumption}[Conditional Ignorability of Treatment Assignment] \label{assumption:ignorability} For some vector of pre-treatment covariates $\widetilde{X} \in \widetilde{\mathcal{X}}$: 
$$Y(1), Y(0) \ \indep \ Z \mid \widetilde{X}.$$
\end{assumption} 
This assumption, also known as \textit{selection on observables}, requires that given pre-treatment covariates $\widetilde{X}$, treatment is `as-if' random. In addition to conditional ignorability, treatment effect estimation generally requires overlap, meaning that all units have a non-zero probability of being treated. 

\begin{assumption}[Overlap] \label{assumption:overlap} For units $i \in 1, ..., n$ and $x \in \widetilde{\mathcal{X}}$, $0 < \Pr(Z = 1 \mid \widetilde{X}= x) < 1.$
\end{assumption} 

Our estimand of interest is the \textit{average treatment effect across the treated} (i.e., ATT): 
$$\tau := \E\left[Y(1) - Y(0) \mid Z = 1\right],$$
where the expectation is taken with respect to the population. The proposed framework can be extended for settings in which researchers are interested in the average treatment effect (ATE), as well as other common missingness settings such as external validity and survey non-response \citep{huang2022sensitivity, hartman2022sensitivity}.

 A common approach to estimating the ATT is by using weighted estimators:
\begin{align} \label{eq:att_est}
  \hat \tau(\hat w) := \frac{1}{\sum_{i=1}^n Z_i} \sum_{i=1}^n Y_i Z_i - \frac{\sum_{i=1}^n \hat w_i Y_i (1-Z_i)}{\sum_{i=1}^n \hat w_i (1-Z_i)}.  
\end{align}
The weights $\hat w_i$ are chosen so that the re-weighted distribution of pre-treatment covariates $\widetilde{X}$ across the control units matches the distribution of $\widetilde{X}$ across the treated units; for example, the population-level inverse propensity score weights  guarantee this. Weights must be estimated. A common approach is to fit a propensity score model to estimate a unit's probability of treatment given the pre-treatment covariates and use the fitted values to construct estimated inverse weights. An alternative approach uses balancing weights, which solve an optimization problem that selects weights to balance sample moments, thereby bypassing a need to estimate a propensity model.   For examples of such balancing weights and additional methods and theory, see \citet{hainmueller2012entropy, zubizarreta2015stable, ben2021balancing}. As shown by \citet{chattopadhyay2021implied}, certain regression estimators can also be represented in the form of weighting estimators if weights are allowed to take on negative values. We note that the theoretical framework introduced in Section \ref{sec:des_sens} can be easily extended to accommodate negative weights, though we focus on more familiar positive-weights settings for ease of exposition. 

When the full set of covariates $\widetilde{X}$ in Assumption \ref{assumption:ignorability} are observed and the weights are correctly specified, the weighted estimator is consistent and unbiased for the ATT. However, in practice, it is impossible to know whether or not all confounders have been measured. Omitted confounders lead to biased estimates. In what follows, we consider the setting in which the full vector of covariates is defined as $\widetilde{X} :=\{X, U\}$, where $X \in \mathcal{X}$ is observed and measured across all units, but $U \in \mathcal{U}$ is unobserved. As such, the estimated weights $\hat w_i$ are functions of  $X$ alone. We assume that the estimated weights $\hat w$ converge in probability to the population weights $w:= \Pr(Z = 1 \mid X)/\Pr(Z = 0 \mid X)$ that condition on $X$ alone. We also define the ideal weights $w^*: = \Pr(Z = 1 \mid X, U)/\Pr(Z = 0 \mid X, U)$ as the population-level inverse propensity score weights in $(X,U)$. Were researchers to use the ideal weights $w^*$, they would consistently recover the ATT. Finally, we define the population weighted estimate $\tau(w) = \mathbb{E}(Y|Z=1)-
{\mathbb{E}(wY|Z=0)}$. We define $\tau(w^*)$ similarly, noting that by construction this quantity is identical to the ATT.

We note that the framework just presented accommodates  balancing weights that converge asymptotically to inverse propensity score weights \citep{ben2021balancing}. Researchers may also relax the assumption of correct specification in settings where specification concerns can be formulated as an omitted variable problem (see \citealp{huang2022sensitivity} and \citealp{hartman2022sensitivity} for more discussion). 

\subsection{Review: Sensitivity Analyses for Weighted Estimators}
Sensitivity analyses allow researchers to assess their study findings' robustness to varying degrees of violation of underlying assumptions. We define a \textit{sensitivity model} as a set of ideal weight vectors $w^*$ over which we are interested in conducting a worst-case analysis. Informally, a sensitivity model consists of all possible $w^*$ within a local neighborhood of the weights ${w}$, where the neighborhood is described by a specified error structure and indexed by a parameter that can be chosen to make the neighborhood larger or smaller. 

The ideal weights $w^*$ give a mapping $w^*(x,u)$ from $\{X, U\}$ to a real-valued weight. Formally, a sensitivity model $\nu(\Gamma, {w})$ is a set of such mappings $\{w^*(x,u): f_\nu(w^*(X,U),{w}(X)) \leq \Gamma\}$, where the function $f_\nu$ measures dissimilarity between two probability distributions (in this case, the distributions induced by the random covariate vector $\mathcal{X}$ under $w(X)$ and  $w^*(X,U)$).
 
$\Gamma \in \mathbb{R}$ is a parameter constraining the overall dissimilarity allowed. Larger $\Gamma$ values allow for larger deviations from ${w}$, and hence, more unobserved confounding.  For any particular choice of $\nu$ and $\Gamma$, we define the interval of possible values for the true ATT:
\begin{equation} 
\left[ L_{\nu(\Gamma, {w})}, U_{\nu(\Gamma, {w})} \right] := \left[ \inf_{\tilde w \in \nu(\Gamma, w)}  \tau(\tilde w), \sup_{\tilde w \in \nu(\Gamma, w)}  \tau(\tilde w)\right].
\label{eqn:partial_id_region} 
\end{equation} 
We refer to this quantity as the \emph{effect interval}.

When assessing whether unobserved confounding is sufficiently strong to overturn a research conclusion, there are two sources of error for which we must account. The interval $\left[ L_{\nu(\Gamma, {w})}, U_{\nu(\Gamma, {w})} \right]$ bounds the first source of error, the bias arising from the omitted confounders. With an infinite number of samples from the population, this would be the only source of error.  However, in practice, we work with estimated weights $\hat{w}$ instead of the population weights $w$, which produce  a noisy approximation $\left[ \hat{L}_{\nu(\Gamma, {w})}, \hat{U}_{\nu(\Gamma, {w})} \right]$ to the effect interval and the resulting sampling variability provides another source of error. Therefore, it is typically necessary to construct a bias-aware confidence interval $CI_{\nu(\Gamma, {w})}(\alpha) \supseteq \left[ \hat{L}_{\nu(\Gamma,{w})}, \hat{U}_{\nu(\Gamma, {w})} \right]$ that contains any true parameter $\tau(\widetilde{w}) \in \left[ L_{\nu(\Gamma, w)}, U_{\nu(\Gamma, w)} \right]$ with probability at least $1-\alpha$.  

Recent literature on sensitivity considers the notion of sharpness, which relates to whether $CI_{\nu(\Gamma, {w})}(\alpha)$ converges in large samples to the effect interval \citep{dorn2021sharp}.\footnote{\citet{dorn2021sharp} do not use the term effect interval; however, their partially identified region is identical to the effect interval for the sensitivity model they consider, and in general the two coincide whenever the partially identified region is an interval.} Of the sensitivity models we consider in detail below, one (the marginal sensitivity model) is known not to be sharp without considering additional constraints, and the other (the variance-based model) has not been proven sharp; as such the specific numeric bounds discussed below will tend to recover a superset of the effect interval in large samples rather than the exact effect interval. However, our results below will apply whether or not a particular method for sensitivity analysis is sharp. 

For a study with a nominally significant result, a \textit{sensitivity analysis} is conducted by searching over values of $\Gamma$, repeating the test for the hardest-to-reject value of $w^*$ in $\nu(\Gamma, w)$, and finding the largest value $\Gamma^*$ for which it is still possible to reject the null.

If $\Gamma^*$ is small, then the initial finding is sensitive to a small amount of unobserved confounding. If $\Gamma^*$ is large, only a strong unobserved confounder could explain the results under a true null hypothesis.

Many different approaches to constructing sensitivity models using different specifications of $f_\nu$ have been proposed, including restricting various $L^p$-norms of the ratio $w^*/w$ and its image under convex functions \citep[e.g.,][]{zhao2019sensitivity, zhang2022bounds, huang2022sensitivity, jin2022sensitivity}. Given the rich and developing literature on different sensitivity models, one key contribution of our proposed method is the flexibility to be applied to \textit{any} sensitivity model that meets a set of relatively weak regularity conditions. However,  for illustrative purposes, we will discuss two common sensitivity models: the variance-based sensitivity model of \citet{huang2022variance} and the marginal sensitivity model of \citet{tan2006distributional}  and \citet{zhao2019sensitivity}.

\subsection{Example: The Variance-based Sensitivity Model}
The variance-based sensitivity model $\nu_{vbm}(R^2, w)$ constrains the variance in $w^*$ not explained by $w$ \citep{huang2022variance}. More formally, for some $R^2 \in [0, 1)$, the variance-based sensitivity model is defined as follows:
\begin{equation}
    \nu_{vbm}(R^2, w) := \left\{ w^*:
    1 \leq \frac{\var(w^* \mid Z = 0)}{\var(w \mid Z = 0)} \leq \frac{1}{1-R^2} \right\}. \label{eq:vbm}
\end{equation}
We characterize the effect interval under the variance-based model via a bias bound, defined for a set $\nu_{vbm}(R^2, w)$ as: 
\begin{align} 
\max&_{\tilde w \in \nu_{vbm}(R^2, w)} \text{Bias}( \tau(w) \mid \tilde w)\nonumber \\ 
&\leq \sqrt{1-\cor(w, Y \mid Z = 0)^2}\sqrt{\frac{R^2}{1-R^2} \var(Y \mid Z = 0) \var(w \mid Z = 0)}
\label{eqn:optim_bias}
\end{align} 
where $\text{Bias}(\tau(w) \mid \tilde w) = \tau(w) - \tau(\tilde w)$. For a fixed $R^2$ value, researchers can estimate the other quantities in Equation \eqref{eqn:optim_bias} using observed sample analogues. We can then construct an interval that contains the effect interval $\left[L_{\nu_{vbm}(R^2, w)}, U_{\nu_{vbm}(R^2, w)}\right]$: 
\begin{equation}
\left[ \tau(w) \pm \sqrt{1-\cor(w, Y \mid Z = 0)^2}\sqrt{\frac{R^2}{1-R^2} \var(Y \mid Z = 0) \var(w \mid Z = 0)} \,\,
\right]. 
\label{eqn:vbm_bias_interval}
\end{equation}

In essence, the variance-based sensitivity model constrains a weighted $L_2$ distance between the ideal weights $w^*$ and the weights $w$ \citep{huang2022variance}, so it is especially relevant in settings in which researchers are comfortable reasoning about the average degree of unobserved confounding across subjects. A percentile bootstrap approach proposed originally in \citet{zhao2019sensitivity} is used in concert with the bias bound to account for sampling variability, creating confidence intervals for the ATT that remain valid even in the presence of confounding under the sensitivity model.

\subsection{Example: The Marginal Sensitivity Model}
\label{subsec:ex_msm}
The marginal sensitivity model $\nu_{msm}(\Lambda, w)$ constrains the worst-case error from omitting a confounder, positing that the ratio between any ideal weight $w^*$ and corresponding $w$ may not exceed $\Lambda \geq 1$ \citep{tan2006distributional, zhao2019sensitivity}. 
\begin{equation}
    \nu_{msm}(\Lambda, w) := \left\{w^*:
     \Lambda^{-1} \leq \frac{w^*}{w} \leq \Lambda 
    \right\}. \label{eq:msm}
\end{equation}

The extrema $[L_{\nu_{msm}(\Lambda, w)}, U_{\nu_{msm}(\Lambda, w)}]$ can be computed using linear programming \citep{zhao2019sensitivity}.  In contrast to the variance-based sensitivity model, the marginal sensitivity model is most useful when the researcher is comfortable reasoning about the maximal degree of confounding for any given subject.  The percentile bootstrap approach of \citet{zhao2019sensitivity} is again used to account for sampling variability.  \cite{dorn2021sharp} introduced alternative approaches to obtain sharp limiting sets of point estimates under the marginal sensitivity model, which require fitting quantile regressions, but we focus on the simpler approach of \citet{zhao2019sensitivity} here.

\subsection{From Sensitivity Analysis to Design Sensitivity}
Sensitivity analyses provide valuable information and recent innovations have improved their interpretability and utility \citep{ding2016sensitivity, Cinelli2020, soriano2021interpretable}.  However, they are underutilized in practice \citep{vanderweele2017sensitivity, hazlett2023unconfounded}. One fundamental drawback is that sensitivity analysis is conducted post-hoc, as a secondary analysis.  If an estimated result is found to be easily overturned by a relatively weak confounder, there is little that the researcher can do. Returning to the analysis and altering the estimation approach to mitigate sensitivity to an omitted confounder after conducting sensitivity analysis may introduce bias; this practice violates the `design principle,’ which forbids consultation of in-sample outcomes during study design \citep{rubin2007design}. 

We propose a design tool, \emph{design sensitivity} for weighted estimators, that enables researchers to plan ahead and tailor studies to improve robustness to unmeasured confounding. In brief, design sensitivity characterizes the power of a sensitivity analysis in large samples. Comparing design sensitivities across different estimation approaches and study specifications provides insight into the implications of those choices for robustness to unmeasured bias, much as power calculations provide insight into design choices' impacts on precision in randomized studies. Design sensitivity in this formal sense was introduced by \citet{rosenbaum2004design} and has been explored extensively for matched studies \citep{heller2009split, rosenbaum2010design, hsu2013effect}. Implications include the value of reducing heterogeneity within matched pairs \citep{rosenbaum2005heterogeneity} and choosing test statistics \citep{rosenbaum2011aspects, howard2021uniform} and  treatment doses \citep{rosenbaum2004design} carefully. However, design sensitivity has, to our knowledge, not been introduced for weighted estimators. This is the focus of our paper. In Section 3, we introduce design sensitivity for weighting estimators. We introduce the notion of power for sensitivity analyses for weighted estimators, which allows us to generally define design sensitivity for a broad class of sensitivity models. We then derive the design sensitivity for both the variance-based sensitivity models and the marginal sensitivity models. In Section 4, we explore the implications of important design choices on design sensitivity.

\section{Design Sensitivity for Weighted Estimators}
\label{sec:des_sens}
\subsection{Power of a Sensitivity Analysis and Design Sensitivity}
 
Design sensitivity is closely related to familiar notions of statistical power.  The power of a test is the probability of rejecting a null hypothesis when a specific alternative is instead true; the alternative is typically chosen to reflect a ``favorable" situation in which the effect of interest is present and important identifying assumptions are met \citep[\S 15]{rosenbaum2010designbook}.  The power of a sensitivity analysis, for a null hypothesis of no treatment effect, a given sensitivity model $\nu(\Gamma, w)$, and a value $\Gamma_0$ of the sensitivity parameter, is the probability that $CI_{\nu(\Gamma_0,w)}$ 
excludes zero under a \textit{favorable} situation, meaning (1) the study is free of unmeasured bias, and in addition (2) a null hypothesis of no treatment effect is false, with potential outcomes generated by a specific stochastic model.

While it may seem paradoxical to evaluate a sensitivity analysis when no unmeasured confounding is present, in practice, researchers do not know whether unmeasured confounding is present and must conduct a sensitivity analysis anyway. Power helps determine when true effects can be detected, even when the test allows for some confounding. Power depends on the potential outcome distribution; as a result, researchers must specify a hypothesized model for the potential outcomes and treatment effects in order to compute power in advance of data analysis. When outcome data is available from a planning sample, 
it can be used to calibrate the hypothetical outcome distribution (see Appendix \ref{app:planning_sample}).

We now provide a convenient normal approximation to the power of a sensitivity analysis, derived originally in \citet{rosenbaum2004design}.  Throughout the paper, we will assume without loss of generality that the specified alternative has a positive treatment effect ($\tau > 0$).

\begin{theorem}[Power of a Sensitivity Analysis ]\label{thm:power}
Let $\hat \tau(\hat w)$ be a standard weighted estimator (i.e., Equation \eqref{eq:att_est}), and for a sensitivity model $\nu(\Gamma, w)$ let 
$\tau_{\nu(\Gamma, w)} := \inf_{\tilde w \in \nu(\Gamma, w)} \tau(\tilde w)$ and $\xi_{\nu(\Gamma, w)} := \tau(w) - \tau_{\nu(\Gamma,w)}$. Finally, let $k_\alpha := 1- \Phi(\alpha)$. Then if $\hat \tau(\hat w)$ is asymptotically normal, the power of a sensitivity analysis is given as follows:  
\begin{align} 
\Pr\left( \frac{\sqrt{n}(\hat \tau(\hat w) - \xi_{\nu(\Gamma, w)}) }{ \sigma_{\nu(\Gamma, w)}} \geq k_\alpha\right) %\nonumber \\
&= \Pr\left( \frac{\sqrt{n} ( \hat \tau(\hat w) - \tau(w))}{ \sigma_{W}} \geq \frac{k_\alpha  \sigma_{\nu(\Gamma, w)} + \sqrt{n} (\xi_{\nu(\Gamma, w)}-\tau(w) )}{\sigma_{W}}\right) \nonumber \\
&\simeq 1 - \Phi \left( \frac{k_\alpha \sigma_{\nu(\Gamma, w)} + \sqrt{n} (\xi_{\nu(\Gamma, w)} - \tau(w))}{\sigma_{W}} \right),
\label{eqn:bias_dist} 
\end{align} 
where $\sigma_{W}$ and $\sigma_{\nu(\Gamma, w)}$ represent the variance of $\tau(w)$ and $\tau_{\nu(\Gamma,w)}$ respectively. 
\end{theorem} 
The $\simeq$ symbol indicates asymptotic equivalence, meaning that the quantities on the two sides converge to the same limit as $n \rightarrow \infty$.  For regularity conditions sufficient to guarantee that $\hat \tau(\hat w)$ is asymptotically normal, see Assumption \ref{assump:regularity_conds} in Appendix \ref{app:pwr_sens}.

Expression  \eqref{eqn:bias_dist} reveals helpful patterns for large $n$.  The numerator of the fractional term includes a variance term that is stable across sample sizes, as well as a bias term that grows with $n$ and will dominate the formula in large samples.  The sign of the bias term $\xi_{\nu(\Gamma, w)} - \tau(w)$ is thus highly consequential, determining whether asymptotic power will be very large or very small. For a given distribution of weights $w$, increasing $\Gamma$ eventually leads to a shift from a high-power to a low-power regime. Design sensitivity characterizes this important phase transition.

\begin{theorem}[Design Sensitivity]\label{thm:ds} 
For a sensitivity model $\nu(\Gamma, w)$ where $\sigma_{\nu(\Gamma, w)} < \infty$ and a given favorable situation, let the \emph{design sensitivity} be any value $\widetilde{\Gamma}$ such that the following two conditions hold, where $\beta_{\nu(\Gamma, w)}$ is the power of the sensitivity analysis:
    \[\beta_{\nu(\Gamma, w)} \longrightarrow 1 \quad \forall\,\, \Gamma < \widetilde{\Gamma} \quad \quad \text{and} \quad \quad 
  \beta_{\nu(\Gamma, w)} \longrightarrow 0 \quad \forall\,\, \Gamma > \widetilde{\Gamma}.\]
Then $\widetilde{\Gamma}$ is given by solving the estimating equation $\xi_{\nu(\Gamma, w)} - \tau(w) = 0$ for $\Gamma$.
\end{theorem} 
The theorem is a direct consequence of Equation \eqref{eqn:bias_dist}. In short, the design sensitivity delineates the maximum amount of unobserved confounding under which the effect can still be detected given a sufficiently large sample. It is often a more useful quantity than the power, which must typically be calculated separately for many values of $\Gamma$ (since a single compelling value of $\Gamma$ cannot often be identified in advance); design sensitivity does not require either a $\Gamma$-value or sample size to be specified.
Through the $\xi_{\nu(\Gamma,w)}$ and $\tau(w)$ terms, design sensitivity depends on the weights selected by the researcher, and evaluating design sensitivities provides a natural basis on which to compare different design specifications.  As will be shown, analytical solutions for the design sensitivity in specific sensitivity models also highlight general principles about the factors that most influence a study's robustness to unobserved confounding.

\subsection{Design Sensitivity in the Variance-based Sensitivity Model}
In the variance-based sensitivity model, the bias term $\xi_{\nu_{vbm}(R^2, w)}$ defined in Theorem  \ref{thm:power} is equal to the bias bound from Equation \eqref{eqn:optim_bias}. This follows immediately from the fact that $\xi_{\nu(\Gamma, w)}$ is generally defined as the difference between the estimates $\tau(w)$ and the minimum value (i.e., $L_{\nu(\Gamma, w)}$). As such, $\xi_{\nu_{vbm}(R^2, w)} := \tau(w) - L_{\nu_{vbm}(R^2, w)} = \tau(w) - \big(\tau(w) - \max_{\tilde w \in \nu_{vbm}(R^2,w)} \text{Bias}(\tau(w) \mid \tilde w) \big) = \max_{\tilde w \in \nu_{vbm}(R^2,w)} \text{Bias}(\tau(w) \mid \tilde w)$. Using Equation \eqref{eqn:optim_bias}, we derive a closed form for the design sensitivity $\widetilde{R}^2$. 
\begin{theorem}[Design Sensitivity for Variance-Based Models]\label{thm:design_sensitivity_vbm} 
Define $\tilde R^2$ as
$$\tilde R^2 := \frac{a^2}{1+a^2} \ \ \ \text{where } a^2 = \frac{1}{1-\cor(w, Y \mid Z = 0)^2} \cdot \frac{ \tau(w)^2}{\var(w \mid Z = 0) \cdot \var(Y \mid Z = 0)},$$
where covariances and variances are computed under the favorable situation.  Then, if the weighted estimator $\hat \tau (\hat w)$ is asymptotically normal, 
\begin{align*}
        & \Pr\left(0 \notin \left[ L_{\nu_{vbm}(R^2, w)}, U_{\nu_{vbm}(R^2, w)} \right] \right) \to 1 \text{ as } n \to \infty \text{ for } R^2 < \Tilde{R^2} \\
       \text{and} \quad & \Pr\left(0 \notin \left[ L_{\nu_{vbm}(R^2, w)}, U_{\nu_{vbm}(R^2, w)} \right] \right) \to 0 \text{ as } n \to \infty \text{ for } R^2 > \Tilde{R^2}.
    \end{align*}
\end{theorem}
Theorem \ref{thm:design_sensitivity_vbm} highlights the drivers of design sensitivity of the variance-based sensitivity model. In particular, the variance of the estimated weights (i.e., $\var(w \mid Z =0)$), the variance of the outcomes (i.e., $\var(Y \mid Z = 0)$), the correlation between the estimated weights and the outcomes (i.e., $\cor(w, Y \mid Z = 0)$), and the effect size will affect the size of the design sensitivity. 

Theorem \ref{thm:design_sensitivity_vbm} suggests that design decisions that reduce the variance in the estimated weights or the variance in the outcomes, or increase the association between the weights and the outcomes, can help increase the design sensitivity, and by extension, improve robustness to unobserved confounding. Examples of such decisions include trimming and augmentation. It is important to note that not all of these design choices guarantee an improvement in design sensitivity because they may affect all three of the components highlighted above. The specific impact of each decision requires numerical assessment, and Section \ref{sec:des_choices} uses design sensitivity calculations to determine when we expect improvements to design sensitivity under each design choice.

\subsection{Design Sensitivity in the Marginal Sensitivity Model} \label{subsec:msm_design} 
In the marginal sensitivity model, unlike the variance-based sensitivity model, design sensitivity does not have a closed form.  We characterize it via an estimating equation. 

\begin{theorem}[Design Sensitivity for the Marginal Sensitivity Model]\label{thm:design_sensitivity_msm}
Define $\Tilde{\Lambda}$ as any solution to the following estimating equation (where $F_{y \mid x}$ represents the population cdf of $y$ given $x$ under the favorable situation): \begin{align*} 
\mathbb{E}&[w Y(0) \mid Z = 0] + \tau \\
&= \underset{\theta \in [0,1]}{\sup} \frac{\Lambda\mathbb{E}\left[w Y(0) G_{\theta}(Y(0)) \mid Z=0 \right]  + \frac{1}{\Lambda}\mathbb{E}\left[ w Y(0) (1-G_{\theta}(Y(0)))\mid Z=0\right]}{ \Lambda\mathbb{E}\left[wG_{\theta}(Y(0)) \mid Z=0 \right] + \frac{1}{\Lambda}\mathbb{E}\left[w (1-G_{\theta}(Y(0))) \mid Z=0 \right]},
\end{align*}
where $G_{\theta}(Y(0)) = \mathbbm{1}\{Y(0) \geq F^{-1}_{Y(0) \mid Z=0}(1-\theta)\}$.

\noindent Then under mild regularity assumptions $\widetilde{\Lambda}$ is the design sensitivity, i.e.
    \begin{align*}
        & \Pr\left(0 \notin \left[L_{\nu_{msm}(\Lambda, w)}, U_{\nu_{msm}(\Lambda, w)} \right]\right) \to 1 \text{ as } n \to \infty \text{ for } \Lambda < \Tilde{\Lambda} \\
       \text{and} \quad & \Pr\left(0 \notin \left[L_{\nu_{msm}(\Lambda, w)}, U_{\nu_{msm}(\Lambda, w)} \right]\right) \to 0 \text{ as } n \to \infty \text{ for } \Lambda > \Tilde{\Lambda}.
    \end{align*}
\end{theorem}

The design sensitivity under the set of marginal sensitivity models is a function of the joint cumulative density function of the weights and the control outcomes, and depends on an optimal cutoff $\theta$. As shown by \citet{zhao2019sensitivity}, the worst-case setting for the marginal sensitivity model requires observations with the smallest $Y(0)$ values to have the smallest weights possible under the model (scaling observed weights by $\Lambda^{-1}$), while setting the observations with largest $Y(0)$ to have the largest weights possible (scaling observed weights by $\Lambda$). 
$\theta$ represents the population version of the cutoff in the ordering at which small weights are replaced by large weights.

The optimal value of $\theta$ (i.e. the value for which the expectation is maximized) depends on $\Lambda$, $w$ and the alternative distribution for the outcomes (conditional on covariates) specified as part of defining the favorable situation. Results similar to Theorem \ref{thm:design_sensitivity_msm} hold for trimmed and augmented weighting estimators (see Appendix \ref{sec:proof_msm_corollary}). 

Because Theorem \ref{thm:design_sensitivity_msm} does not give a closed expression for $\widetilde{\Lambda}$, we examine the drivers of design sensitivity for the marginal sensitivity model using numerical simulations. 
\begin{example}[Drivers of Design Sensitivity]\label{ex:sim} 
Define the treatment assignment process and the outcome model as follows: 
$$P(Z = 1 \mid X) \propto \frac{\exp(\beta_\pi X)}{1+\exp(\beta_\pi X)}, \ \ \ \ \ \ Y = \beta_y X + \tau Z + u,$$
where $X \iid N(\mu_x, \sigma^2_x)$ and $u \iid N(0, \sigma^2_y)$. We sample from this process under different choices for parameters $\{\beta_\pi, \beta_y, \tau, \sigma_y\}$ and estimate  design sensitivity under both the variance-based sensitivity model and the marginal sensitivity model (see Appendix \ref{app:sim} for full details). In contrast to the variance-based sensitivity model, design sensitivity under the marginal sensitivity model is primarily driven by the effect size and the variance in the outcomes. 
\begin{figure}[!ht]
\centering 
\includegraphics[width=0.8\textwidth]{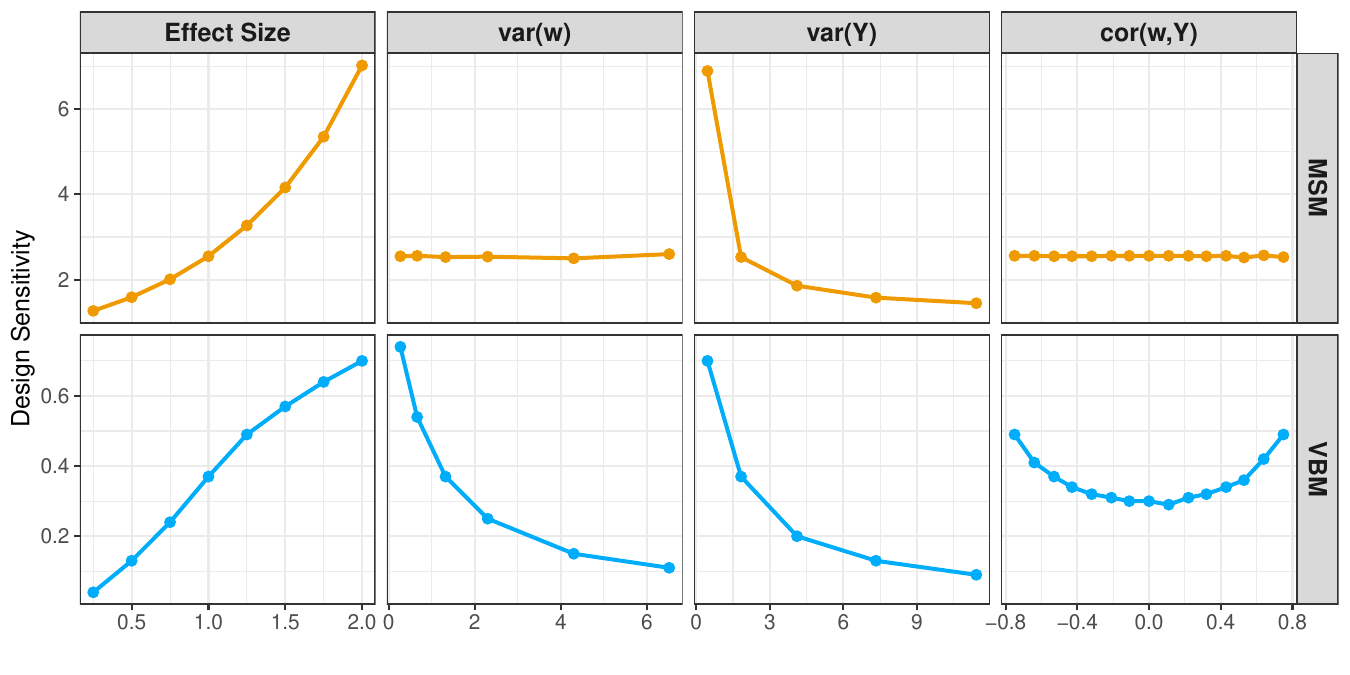}
\caption{Magnitude of the design sensitivities for both MSM and VBM varying different aspects of the data generating process.} 
\label{fig:ex_1}
\end{figure} 
\end{example}

In the following section, we show how two specific design choices --- augmentation of weighting designs with outcome models and trimming of estimated weights --- influence design sensitivity, providing new perspective on the advantages they bring.

\section{Design Choices that Impact Design Sensitivity} \label{sec:des_choices}

\subsection{Augmentation using Outcome Models}
Consider an augmented weighted estimator of the following form: 
$$\hat \tau^{aug}(\hat w)= \frac{1}{\sum_{i=1}^n Z_i} \sum_{i=1}^n Z_i Y_i - \left(\frac{\sum_{i=1}^n \hat w_i (1-Z_i) (Y_i - \hat g(X_i))}{\sum_{i=1}^n \hat w_i (1-Z_i)} + \frac{1}{\sum_{i=1}^n Z_i} \sum_{i=1}^n Z_i \hat g(X_i)\right),$$
where $\hat g(X_i)$ represents an estimated outcome model. The augmented weighted estimator is doubly robust: as long as either the outcome model (i.e., $\hat g$), or the treatment assignment model (i.e., $\hat w_i$) is correctly specified, the augmented weighted estimator will consistently estimate the ATT \citep{tan2007comment, bang2005doubly, kang2007demystifying}. 
However, doubly robust estimation does not eliminate concerns about omitted variable bias. More specifically, if there is an omitted confounder that is relevant to both the treatment and outcome, then neither the outcome model nor the treatment assignment model will be correctly specified. Compounding these concerns, when one (or both) of the models is misspecified, the finite-sample performance of augmented weighted estimators may be inferior to standard weighted (or regression) estimators \citep{kang2007demystifying}. 

Our key result helps resolve these questions by demonstrating that augmentation can improve robustness to unobserved confounding in large samples even when the outcome model is misspecified, suggesting a clear advantage to augmented estimation distinct from double robustness.

\begin{theorem}[Impact of Augmentation on Design Sensitivity] \label{thm:aug} \mbox{}\\
Let $e := Y -  g(X)$ be the population residual from an arbitrary, fixed outcome model $g$ used to augment a weighted estimate. Then, for the variance-based sensitivity model, the design sensitivity from an augmented weighted estimator will be greater than the design sensitivity for a standard weighted estimator if the following holds: 
\begin{equation} 
\var(e \mid Z = 0) \leq \frac{1-\cor(w, Y \mid Z = 0 )^2}{1-\cor(w, e \mid Z = 0)^2} \cdot \var(Y \mid Z = 0).
\label{eqn:aug} 
\end{equation} 
\end{theorem} 
While Theorem \ref{thm:aug} assumes a fixed outcome model $g$, we may relax the assumption of a fixed outcome model by extending the $Z$-estimation framework, introduced in 
\citet{zhao2019sensitivity}.

If the correlation between the estimated weights and the outcomes and the correlation between the estimated weights and the residuals are roughly similar (i.e., $\cor(w, Y \mid Z = 0) \approx \cor(w, e \mid Z = 0)$), then Equation \eqref{eqn:aug} simplifies to a simple comparison between the variance across the residuals and the variance of the outcomes (i.e., $\var(e \mid Z = 0) \leq \var(Y \mid Z = 0)$). 

Theorem \ref{thm:aug} highlights that the degree of improved robustness from augmentation depends directly on how much variation the estimated outcome model is able to explain in the outcomes across the control group. In other words, if $\var(e \mid Z = 0)$ is relatively small, then we expect a larger improvement in design sensitivity from augmentation. Importantly, the gains to design sensitivity from augmentation are not dependent on any additional specification assumptions. Even if the outcome model is misspecified, if it successfully explains variance in the control outcomes (while maintaining similar outcome-weight correlations), then augmentation will improve the robustness of estimated effects.

Similar results hold for the marginal sensitivity model, although we cannot obtain closed-form criteria (see Appendix \ref{sec:app_proofs} for more discussion). This is consistent with the results from Example \ref{ex:sim}, in which the design sensitivity of the marginal sensitivity model varied systematically with the variation in the outcomes. Highly variable outcomes lead to more extreme outcome values and worse worst-case bounds as in the matching context described by \citet{rosenbaum2005heterogeneity}, so stabilizing outcomes improves robustness by limiting the extremity of worst-case settings.

\paragraph{Remark.} For the marginal sensitivity model, \citet{dorn2021sharp} showed that although the usual sensitivity analysis (described in Section \ref{subsec:ex_msm}) does not generally produce sharp bounds, augmentation using a correctly-specified outcome model can render the sensitivity bounds sharp in some settings. This provides another possible explanation for the improvements in design sensitivity arising from augmentation: at least in part, augmentation may be eliminating looseness in the bounds on the effect interval.  This suggests that augmented versions of  provably sharp methods for sensitivity analysis, such as \citet{dorn2021sharp}'s approach based on quantile regression, may not deliver the improvements in design sensitivity that we see here for our two approaches.

\subsection{Trimming}
\label{sec:des_trimming}
Another design decision that commonly arises in practice is trimming, or exclusion of units with extreme weights. Trimming implicitly redefines the estimand of interest to exclude units with extreme propensity scores, which can be helpful in cases in which researchers are worried about potential overlap or positivity violations. Under trimming, we consider a modified estimand: 
\begin{equation} 
\tau_{trim} := \E(Y(1) - Y(0) \mid Z = 1, X \in \mathcal{A}),
\label{eqn:trim_estimand} 
\end{equation} 
where $\mathcal{A}:= \{x \in X \mid a \leq P(Z = 1 \mid x) \leq 1-a\}$—i.e., the set of covariate values for which the probability of treatment, conditional on the observed covariates, is strictly bounded away from 0 and 1 \citep{crump2009dealing}. Equation \eqref{eqn:trim_estimand} defines the estimand as a function of conditional probabilities in the \textit{observable} covariates $X$. When researchers wish to consider trimming with respect to an oracle set $\mathcal{A}^*$, which also conditions on the omitted variable, the underlying procedure for estimating the bounds for both sensitivity models must be changed to account for trimming with respect to the ideal weights $w^*$ instead of $w$. We defer the details of such a procedure for future work. 

In practice, in the context of ATT estimation, we focus on trimming large weights by choosing a cutoff $m$ for the weights $w$. The following theorem shows that for any degree of trimming, the relative improvement to design sensitivity for the variance-based sensitivity model depends on how successfully trimming reduces the variance of the weights compared to the reduction in the variance of $Y$ and the change in the correlation between weights and outcomes.  Design sensitivity for the marginal sensitivity model under trimming is described in Appendix \ref{sec:proof_msm_corollary}. 

\begin{theorem}[Impact of Trimming Weights on Design Sensitivity] \label{thm:trim} 
Let $m$ be a cutoff above which weights are trimmed. Assume that the trimmed weights have mean 1 and that the treatment effect is constant. Then, for the variance-based sensitivity model, the design sensitivity from a trimmed estimator will be greater than the standard weighted estimator if the following holds:
\begin{equation} 
\underbrace{\frac{\var(w \mid w < m, Z = 0)}{\var(w\mid Z =0 )}}_{(a) \text{Variance reduction in } w} \leq \underbrace{\frac{1-\cor(w, Y \mid Z = 0)^2}{1-\cor(w, Y \mid w < m, Z = 0)^2}}_{(b) \text{ Change in relationship between } w \text{ and } Y} \cdot \underbrace{\frac{\var(Y \mid Z = 0)}{\var(Y \mid w < m, Z = 0)}}_{(c) \text{ Variance reduction in } Y}.
\label{eqn:trimming_criteria} 
\end{equation} 
\end{theorem} 
Unlike augmentation, in which design sensitivity is improved so long as researchers are able to estimate an outcome model that explains some variation in the outcome, trimming provides weaker guarantees on improvements to design sensitivity. More specifically, Equation \eqref{eqn:trimming_criteria} provides a bound on the necessary variance reduction in the weights to improve design sensitivity. By construction, the variance of the trimmed weights (i.e., $\var(w \mid w < m, Z = 0)$) will be no greater than the variance of the untrimmed weights (i.e., $\var(w \mid Z = 0)$). If the right-hand side of Equation \eqref{eqn:trimming_criteria} were greater than or equal to 1, then this bound would be trivially met. 

The magnitude of the bound depends on the correlation between the weights and the outcome. For intuition, consider the scenario in which the weights and the outcomes are highly correlated. Removing extreme weights will also remove extreme outcomes, such that the post-trimming outcome variance will be smaller than the initial variance. As a result, we expect $\var(Y \mid w < m, Z = 0)$ to be less than $\var(Y \mid Z = 0)$, thereby increasing term (c) in Equation \eqref{eqn:trimming_criteria}. In cases when the weights are completely unrelated to the outcome,  we expect $\var(Y \mid Z = 0) \approx \var(Y \mid w < m, Z = 0)$. In that case, Equation \eqref{eqn:trimming_criteria}-(c) will be approximately equal to 1. 

An added complexity is that the bound is also dependent on potential changes in the linear relationship between $w$ and $Y$ (i.e., Equation \eqref{eqn:trimming_criteria}-(b)). As a result, if extreme values in the weights correspond to large values of the outcome $Y$, then by trimming, we may reduce the correlation between $w$ and $Y$. In practice, $\cor(w, Y \mid Z = 0)$ tends to be relatively low; as a result, we expect changes in the relationship between $w$ and $Y$ from trimming to be relatively small. 

Theorem \ref{thm:trim} assumed a constant treatment effect to simplify the criteria in Equation \ref{eqn:trimming_criteria}. Notably, this assumption is not necessary for the existence of design sensitivity for the trimmed estimator (see Appendix \ref{sec:app_proofs} for more discussion). However, the existence of treatment effect heterogeneity introduces additional complexity for evaluating the impact of trimming on design sensitivity. If weights and treatment effects are positively correlated, trimming large weights will tend to exclude the subjects with the largest treatment effects. As such, trimming would reduce the treatment effect size and therefore also reduce design sensitivity. Conversely, trimming can increase the effect size and improve design sensitivity when the weights and treatment effect are negatively correlated. For more discussion of the connection between treatment effect heterogeneity and design sensitivity, see \citet{rosenbaum2007confidence}.

 In practice, researchers may utilize Theorem \ref{thm:design_sensitivity_vbm} and Theorem \ref{thm:design_sensitivity_msm} to estimate the design sensitivity under different trimming criteria. While Theorem \ref{thm:trim} is formulated with respect to a trimmed estimator that omits units with extreme weights, the results easily extend in cases when researchers use a smooth trimmed estimator instead (see Appendix \ref{sec:app_proofs} for details). By computing design sensitivity, researchers can directly assess whether or not trimming can help improve robustness to omitted variable bias, and whether or not these potential gains to robustness are worth the trade-off of using a different estimand of interest. For example, if researchers estimate the design sensitivity under trimming and find that for a small effect size, trimming a small portion of extreme weights results in a large improvement in design sensitivity, it may be helpful to perform trimming. In contrast, if researchers find that only at a large effect size or only by trimming large number of weights would trimming help with design sensitivity, it would not be worth performing trimming and altering the estimand of interest. 

\section{Empirical Application: FARC Peace Agreement}
\label{sec:empirical_app}
\subsection{Background and Context}
To illustrate design sensitivity, we re-analyze a study from \cite{hazlett2023unconfounded}. After decades of fighting, the Colombian government under President Juan Santos reached a historic peace deal with the Revolutionary Armed Forces of Colombia (FARC). However, in a 2016 referendum, the public narrowly voted to reject the peace deal. The FARC peace deal remains an important case study in understanding drivers of support for peace. 

Following \cite{hazlett2023unconfounded}, we examine two prevalent hypotheses for drivers of support for peace: (1) exposure to violence, and (2) presidential support. We define the outcome of interest as the proportion of individuals at the municipality level who voted in favor of the peace deal. For exposure to violence, treatment is defined by whether any recorded deaths attributed to FARC occurred in a municipality. For presidential support, treatment is defined as whether or not President Santos won the popular vote in the municipality in the second round of presidential elections (which would have implied that he won that particular region). We estimate inverse propensity score weights by fitting a logistic regression using the available pre-treatment covariate data. This includes variables such as past incidents of FARC-related deaths, GDP per capita for a specific municipality, and the number of people who live in each municipality. 

We vary the possible treatment effects and estimate the resulting design sensitivities. Figure \ref{fig:farc} displays the results. Because we are examining the percentage of individuals who vote in favor of the peace deal in a municipality as the outcome, the range of possible treatment effects is restricted by the fact that the average treatment outcome cannot be outside the range 0 - 100\%. In practice, researchers can estimate design sensitivity by calibrating to a chosen outcome distribution. For illustrative purposes, we calibrate the estimated design sensitivities using the true outcome distribution. However, the results can be estimated for any arbitrary distribution for the outcome. See Appendix \ref{app:planning_sample} for recommendations for using a planning sample to calibrate outcome distributions.

\subsection{Illustration on the Variance-Based Sensitivity Models}
To examine the potential impact of augmentation, we vary the amount of variation that can be explained in the control outcomes by a hypothetical outcome model and estimate the updated design sensitivity. Consistent with Theorem \ref{thm:aug}, we see that as the amount of variation explained in the outcomes increases, the amount of improvement in design sensitivity also increases. Notably, even in cases when the outcome model can explain 50\% of the variation in the control outcomes, there is a relatively limited impact on the design sensitivity for the variance-based model. 

To assess the impact of trimming, we estimate the design sensitivity of the weighted estimator, trimming at thresholds of 0.9 and 0.8 (trimming observations with estimated propensity scores greater than these values). This improves the design sensitivity uniformly across all effect sizes. Trimming a small number of extreme weights results in large improvements in the design sensitivity, even for a relatively low effect size. For example, for the hypothesis of exposure to violence, trimming weights that correspond to propensity scores greater than 0.9 would result in excluding 3 observations, out of 1,123 total observations (i.e., 0.26\% of control units). For an effect size of 10, this would correspond to an increase in the design sensitivity from $\tilde R^2 = 0.01$ to $\tilde R^2 = 0.11$. This implies that assuming an effect size of 10, without trimming, we could only identify a true effect if the imbalance from an omitted confounder explained less than 1\% of the variation in the ideal weights. However, after trimming, we would be able to identify a true effect even accounting for a possible confounder that is up to 10 times more imbalanced. 

Importantly, neither design choice (i.e., augmenting and trimming) appear to hurt design sensitivity for the variance-based sensitivity models in this setting. However, from estimating the design sensitivities, it is clear that there are substantial gains to robustness from trimming a few observations from the study. In contrast, while fitting a predictive outcome model can help improve design sensitivity, these improvements are more marginal. 

\begin{figure}[!t]
\centering 
\includegraphics[width=\textwidth]{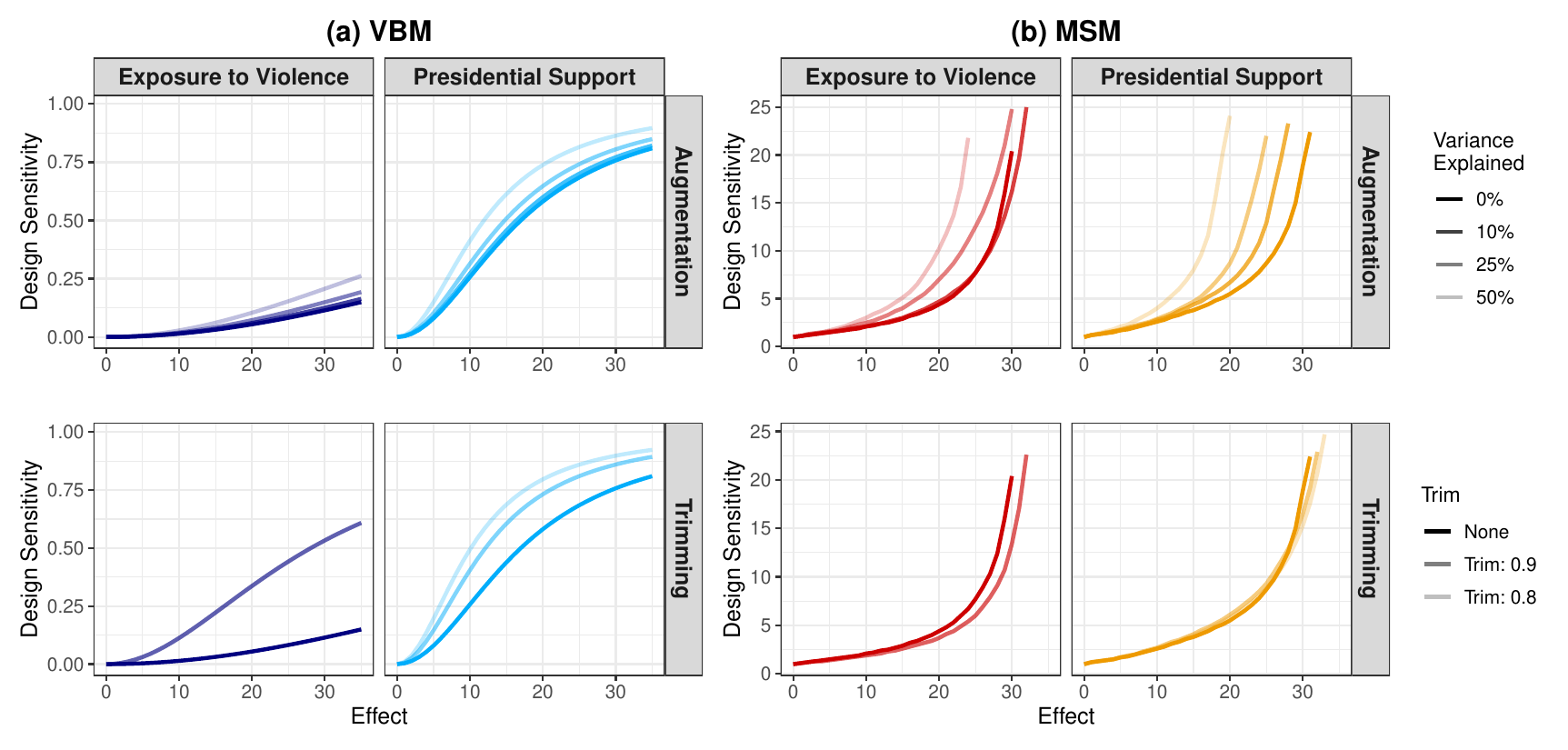}
\caption{Design sensitivities under augmentation and trimming for (a) the variance-based model and (b) the marginal sensitivity model.} 
\label{fig:farc} 
\end{figure} 

\subsection{Illustration on the Marginal Sensitivity Model}
We now illustrate design sensitivity on the marginal sensitivity model. We see a large improvement in the design sensitivity for the marginal sensitivity model from augmentation across all effect sizes. This is likely because the marginal sensitivity model is more susceptible to extreme values in the outcome. By augmenting the weighted estimator, we are able to reduce extreme values in the outcome. Even an outcome model that can explain 10\% of the variation in the control outcomes can result in a substantial improvement in the design sensitivity of the marginal sensitivity model. (See Figure \ref{fig:farc}-(b) for illustration.)

In contrast to the variance-based sensitivity models, we see that for for the marginal sensitivity models, at small effect sizes, trimming does not affect the design sensitivity. However, for large effect sizes (i.e., $\tau > 25$), trimming actually results in a slight reduction in the design sensitivity. As such, in settings where researchers are concerned about robustness to a worst-case error, design sensitivity would suggest that researchers should not perform trimming, and fitting a predictive outcome model would be most helpful at improving robustness.

\subsection{Using a Planning Sample to Improve Power}
\label{sec:power_farc}

In practice, we cannot use the full data to estimate the design sensitivities, as this would violate the design principle. 
\citet{heller2009split} propose randomly splitting the data in an observational study into a planning sample and an analysis sample to inform design decisions. However, using a planning sample comes at a cost; in particular, if we hold out part of the data to use as a planning sample, we cannot use these observations in our analysis (i.e., we restrict ourselves to a smaller sample size). We consider using a similar strategy for weighted observational studies to select between traditional inverse propensity score (IPW), trimmed, and augmented weighting estimators, and illustrate that even with sample splitting, we can improve power for the FARC example. 

To estimate power, we randomly split the data into planning and analysis samples and use the planning sample to estimate design sensitivities for each method and both sensitivity models. (See Appendix \ref{app:planning_sample} for details.) We then conduct sensitivity analyses for the variance-based and marginal sensitivity models under particular values of $R^2$ and $\Lambda$, respectively, using the selected estimation strategies and record whether or not the null hypothesis of no treatment effect is rejected at a 5\% significance level. We repeat this process for 1,000 random splits of the data, estimating power as the proportion of random splits for which we reject.

We estimate the power for each combination of estimator and sensitivity model when using 10\% of the FARC data with the presidential support treatment for the planning sample and the remaining 90\% for the analysis sample. The results are available in Table \ref{tab:est_power10}. For the marginal sensitivity model with $\Lambda = 4$ and the variance-based sensitivity model with $R^2 = 0.25$, the sensitivity analysis for each estimator rejects the null hypothesis of no effect. Additionally, implementing the method selected using the planning sample achieves near or equal to 100\% power in both scenarios. Conversely, the power is near zero for each estimator with $\Lambda = 6$, rendering the choice of estimator moot. The repeated sample splits with $\Lambda = 5$ for the marginal sensitivity model and $R^2 = 0.35$ and $0.67$ for the variance-based model highlight the potential gains in power from using a planning sample to make design decisions. For the former sensitivity model, the augmented weighting estimator greatly outperforms the two alternative estimators. Implementing the method selected by the sample splitting approach yields 68\% power, far higher than would be achieved by using the IPW or trimmed weighting estimator. For $R^2 = 0.35$ and $R^2 = 0.67$, using a planning sample leads us to select the trimmed estimator for each sample split, maximizing power. We repeat the same exercise using 20\% of the data for the planning sample. 

\begin{table}[!ht]
\centering
\textbf{Estimated power for analysis sample using FARC data} \\ 
\vspace{2mm} 
\resizebox{\textwidth}{!}{%
\begin{tabular}{lccccccccccc}
\toprule
& \multicolumn{3}{c}{Reject?} & \multicolumn{4}{c}{10\% of Data for Planning} & \multicolumn{4}{c}{20\% of Data for Planning} \\ \cmidrule(lr){2-4} \cmidrule(lr){5-8} \cmidrule(lr){9-12}
& \multicolumn{3}{c}{Full Sample} & \multicolumn{3}{c}{Analysis Sample} & Chosen &\multicolumn{3}{c}{Analysis Sample} & Chosen\\ 
 &  IPW & Trim  & Aug. &  IPW & Trim  & Aug. & Method &  IPW & Trim  & Aug. &  Method\\ \midrule
% \multirow{3}{*}{MSM} &
\multicolumn{7}{l}{Marginal Sensitivity Model} \\ \midrule
$\Lambda = 4$ & 1 & 1 & 1 & 0.97 & 1.00 & 1.00 & 0.99 & 0.84 & 1.00 & 1.00 & 0.93 \\
$\Lambda = 5$ & 0 & 0 & 1 & 0.10 & 0.28 & 0.86 & 0.68 & 0.18 & 0.40 & 0.73 & 0.51 \\
$\Lambda = 6$ & 0 & 0 & 0 & 0.01 & 0.00 & 0.02 & 0.02 & 0.03 & 0.02 & 0.06 & 0.04 \\
\midrule
\multicolumn{7}{l}{Variance-based Sensitivity Model} \\\midrule
$R^2 = .25$ & 1 & 1 & 1 & 0.81 & 1.00 & 0.85 & 1.00 & 0.67 & 1.00 & 0.73 & 1.00\\
$R^2 = .35$ & 0 & 1 & 0 & 0.18 & 1.00 & 0.27 & 1.00 & 0.22 & 1.00 & 0.28 & 1.00 \\
$R^2 = .67$ & 0 & 0 & 0 & 0.00 & 0.79 & 0.00 & 0.79 & 0.00 & 0.96 & 0.00 & 0.96  \\
\bottomrule
\end{tabular}
}
\caption{For values under full sample, 1 represents rejection of the null hypothesis of no effect at a 95\% significance level for the corresponding estimator and sensitivity parameter value using the full FARC data, while 0 represents failure to reject. Under analysis sample, we display the proportion of rejections across repeated splits of the data as estimated power. Chosen method is the estimated power using the method selected using the planning sample.}
\label{tab:est_power10} 
\end{table}

\subsection{Choice of Sensitivity Model and Interpretation} 
While design sensitivity is useful for deciding which weighting approaches we expect to be most robust given a particular mode of sensitivity analysis, we do not view it as particularly helpful for choosing among different sensitivity models. Design sensitivities for the marginal and variance-based models are parameterized very differently and are not directly comparable.

As such, throughout the paper and analysis, we have restricted attention to maximizing the design sensitivity for a \textit{fixed} set of sensitivity models. We refer readers to \cite{huang2022variance} for more discussion about comparing the performance of different sensitivity models, as well as \cite{rosenbaum2015bahadur} for discussion on comparisons of different sensitivity models below the design sensitivity threshold. In general, we recommend that practitioners decide \textit{a priori} which sensitivity model best captures the type of unobserved confounder that concerns them most from a substantive perspective and pursue design sensitivity under that model. 

The interpretation of design sensitivity magnitude depends on the underlying sensitivity model. As such, determining whether a design sensitivity is large or small requires researchers to reason about whether the amount of confounding represented by $\tilde \Gamma$ is plausible. We recommend the use of existing tools for sensitivity analysis such as formal benchmarking \citep[e.g.,][]{Cinelli2020, huang2022variance, huang2022sensitivity} and amplification  \citep{soriano2021interpretable} to help with interpretation. 

\section{Conclusion}  
\label{sec:concl}
In this paper, we introduced \textit{design sensitivity} for weighted estimators. This asymptotic measure of robustness allows researchers to consider how certain design choices in their observational studies can affect sensitivity to omitted confounders at the design stage. Design sensitivity can be estimated for a general set of sensitivity models that meet a relatively weak set of regularity conditions. We derive the design sensitivities for two commonly used sensitivity models: the variance-based sensitivity model, which constrains an average error from omitting a confounder, and the marginal sensitivity model, which constrains a worst-case error from omitting a confounder. We show that trimming and augmentation—two common design choices that researchers make in practice—can influence design sensitivity. Thus, beyond the standard discussions of variance reduction for trimming and double robustness for augmentation, these decisions can also impact robustness to omitted confounders. We illustrate our framework in a study of the 2016 Colombian peace agreement, for which trimming drastically improves design sensitivity under the variance-based sensitivity model, and augmentation improves design sensitivity under the marginal sensitivity model. 

Several lines of future work follow naturally. While we provide explicit calculations for two commonly used sensitivity models, the framework introduced applies more generally. An interesting avenue of future work could compare how design sensitivities across a wider array of different sensitivity models respond to design choices. Future work could also examine design sensitivity with respect to other design choices, including the choice between ATT, ATE, and quantile effect estimands \citep{greifer2021choosing}.

Finally, design sensitivity is defined in the context of the asymptotic limiting distributions. This formulation allows us to disentangle uncertainty from sampling error from uncertainty from omitted confounding.  However, it may not provide clear guidance in small-sample settings or in settings where several designs exhibit similar design sensitivities. In the context of matching, \citet{rosenbaum2015bahadur} computes Bahadur efficiencies of sensitivity analyses as a way to compare robustness of different design specifications with more granularity than design sensitivity can provide. Future work should explore the potential of this approach for weighting estimators. 
%\clearpage
%BIBLIOGRAPHY: 
\bibliographystyle{chicago} 
\bibliography{references}

\begin{thebibliography}{}

\bibitem[\protect\citeauthoryear{Bang and Robins}{Bang and
  Robins}{2005}]{bang2005doubly}
Bang, H. and J.~M. Robins (2005).
\newblock Doubly robust estimation in missing data and causal inference models.
\newblock {\em Biometrics\/}~{\em 61\/}(4), 962--973.

\bibitem[\protect\citeauthoryear{Ben-Michael, Feller, Hirshberg, and
  Zubizarreta}{Ben-Michael et~al.}{2021}]{ben2021balancing}
Ben-Michael, E., A.~Feller, D.~A. Hirshberg, and J.~R. Zubizarreta (2021).
\newblock The balancing act in causal inference.
\newblock {\em arXiv preprint arXiv:2110.14831\/}.

\bibitem[\protect\citeauthoryear{Chattopadhyay and Zubizarreta}{Chattopadhyay
  and Zubizarreta}{2021}]{chattopadhyay2021implied}
Chattopadhyay, A. and J.~R. Zubizarreta (2021).
\newblock On the implied weights of linear regression for causal inference.
\newblock {\em arXiv preprint arXiv:2104.06581\/}.

\bibitem[\protect\citeauthoryear{Cinelli and Hazlett}{Cinelli and
  Hazlett}{2020}]{Cinelli2020}
Cinelli, C. and C.~Hazlett (2020).
\newblock {Making Sense of Sensitivity: Extending Omitted Variable Bias}.
\newblock {\em Journal of the Royal Statistical Society Series B\/}~{\em
  82\/}(1), 39--67.

\bibitem[\protect\citeauthoryear{Crump, Hotz, Imbens, and Mitnik}{Crump
  et~al.}{2009}]{crump2009dealing}
Crump, R.~K., V.~J. Hotz, G.~W. Imbens, and O.~A. Mitnik (2009).
\newblock Dealing with limited overlap in estimation of average treatment
  effects.
\newblock {\em Biometrika\/}~{\em 96\/}(1), 187--199.

\bibitem[\protect\citeauthoryear{Ding and VanderWeele}{Ding and
  VanderWeele}{2016}]{ding2016sensitivity}
Ding, P. and T.~J. VanderWeele (2016).
\newblock Sensitivity analysis without assumptions.
\newblock {\em Epidemiology (Cambridge, Mass.)\/}~{\em 27\/}(3), 368.

\bibitem[\protect\citeauthoryear{Dorn and Guo}{Dorn and
  Guo}{2021}]{dorn2021sharp}
Dorn, J. and K.~Guo (2021).
\newblock Sharp sensitivity analysis for inverse propensity weighting via
  quantile balancing.
\newblock {\em arXiv preprint arXiv:2102.04543\/}.

\bibitem[\protect\citeauthoryear{Greifer and Stuart}{Greifer and
  Stuart}{2021}]{greifer2021choosing}
Greifer, N. and E.~A. Stuart (2021).
\newblock Choosing the estimand when matching or weighting in observational
  studies.
\newblock {\em arXiv preprint arXiv:2106.10577\/}.

\bibitem[\protect\citeauthoryear{Hainmueller}{Hainmueller}{2012}]{hainmueller2012entropy}
Hainmueller, J. (2012).
\newblock Entropy balancing for causal effects: A multivariate reweighting
  method to produce balanced samples in observational studies.
\newblock {\em Political analysis\/}~{\em 20\/}(1), 25--46.

\bibitem[\protect\citeauthoryear{H{\'a}jek}{H{\'a}jek}{1971}]{hajek1971comment}
H{\'a}jek, J. (1971).
\newblock Comment on “an essay on the logical foundations of survey sampling,
  part one”.
\newblock {\em The foundations of survey sampling\/}~{\em 236}.

\bibitem[\protect\citeauthoryear{Hartman and Huang}{Hartman and
  Huang}{2022}]{hartman2022sensitivity}
Hartman, E. and M.~Huang (2022).
\newblock Sensitivity analysis for survey weights.
\newblock {\em arXiv preprint arXiv:2206.07119\/}.

\bibitem[\protect\citeauthoryear{Hazlett and Parente}{Hazlett and
  Parente}{2023}]{hazlett2023unconfounded}
Hazlett, C. and F.~Parente (2023).
\newblock From" is it unconfounded?" to" how much confounding would it take?":
  Applying the sensitivity-based approach to assess causes of support for peace
  in colombia.
\newblock {\em Journal of Politics (forthcoming)\/}.

\bibitem[\protect\citeauthoryear{Heller, Rosenbaum, and Small}{Heller
  et~al.}{2009}]{heller2009split}
Heller, R., P.~R. Rosenbaum, and D.~S. Small (2009).
\newblock Split samples and design sensitivity in observational studies.
\newblock {\em Journal of the American Statistical Association\/}~{\em
  104\/}(487), 1090--1101.

\bibitem[\protect\citeauthoryear{Hirshberg, Maleki, and Zubizarreta}{Hirshberg
  et~al.}{2019}]{hirshberg2019minimax}
Hirshberg, D.~A., A.~Maleki, and J.~R. Zubizarreta (2019).
\newblock Minimax linear estimation of the retargeted mean.
\newblock {\em arXiv preprint arXiv:1901.10296\/}.

\bibitem[\protect\citeauthoryear{Howard and Pimentel}{Howard and
  Pimentel}{2021}]{howard2021uniform}
Howard, S.~R. and S.~D. Pimentel (2021).
\newblock The uniform general signed rank test and its design sensitivity.
\newblock {\em Biometrika\/}~{\em 108\/}(2), 381--396.

\bibitem[\protect\citeauthoryear{Hsu, Small, and Rosenbaum}{Hsu
  et~al.}{2013}]{hsu2013effect}
Hsu, J.~Y., D.~S. Small, and P.~R. Rosenbaum (2013).
\newblock Effect modification and design sensitivity in observational studies.
\newblock {\em Journal of the American Statistical Association\/}~{\em
  108\/}(501), 135--148.

\bibitem[\protect\citeauthoryear{Huang}{Huang}{2022}]{huang2022sensitivity}
Huang, M. (2022).
\newblock Sensitivity analysis in the generalization of experimental results.
\newblock {\em arXiv preprint arXiv:2202.03408\/}.

\bibitem[\protect\citeauthoryear{Huang and Pimentel}{Huang and
  Pimentel}{2022}]{huang2022variance}
Huang, M. and S.~D. Pimentel (2022).
\newblock Variance-based sensitivity analysis for weighting estimators result
  in more informative bounds.
\newblock {\em arXiv preprint arXiv:2208.01691\/}.

\bibitem[\protect\citeauthoryear{Ishikawa and He}{Ishikawa and
  He}{2023}]{ishikawa2023kernel}
Ishikawa, K. and N.~He (2023).
\newblock Kernel conditional moment constraints for confounding robust
  inference.
\newblock {\em arXiv preprint arXiv:2302.13348\/}.

\bibitem[\protect\citeauthoryear{Jin, Ren, and Zhou}{Jin
  et~al.}{2022}]{jin2022sensitivity}
Jin, Y., Z.~Ren, and Z.~Zhou (2022).
\newblock Sensitivity analysis under the $ f $-sensitivity models: Definition,
  estimation and inference.
\newblock {\em arXiv preprint arXiv:2203.04373\/}.

\bibitem[\protect\citeauthoryear{Kang and Schafer}{Kang and
  Schafer}{2007}]{kang2007demystifying}
Kang, J.~D. and J.~L. Schafer (2007).
\newblock Demystifying double robustness: A comparison of alternative
  strategies for estimating a population mean from incomplete data.

\bibitem[\protect\citeauthoryear{Rosenbaum}{Rosenbaum}{2004}]{rosenbaum2004design}
Rosenbaum, P.~R. (2004).
\newblock Design sensitivity in observational studies.
\newblock {\em Biometrika\/}~{\em 91\/}(1), 153--164.

\bibitem[\protect\citeauthoryear{Rosenbaum}{Rosenbaum}{2005}]{rosenbaum2005heterogeneity}
Rosenbaum, P.~R. (2005).
\newblock Heterogeneity and causality: Unit heterogeneity and design
  sensitivity in observational studies.
\newblock {\em The American Statistician\/}~{\em 59\/}(2), 147--152.

\bibitem[\protect\citeauthoryear{Rosenbaum}{Rosenbaum}{2007}]{rosenbaum2007confidence}
Rosenbaum, P.~R. (2007).
\newblock Confidence intervals for uncommon but dramatic responses to
  treatment.
\newblock {\em Biometrics\/}~{\em 63\/}(4), 1164--1171.

\bibitem[\protect\citeauthoryear{Rosenbaum}{Rosenbaum}{2010a}]{rosenbaum2010designbook}
Rosenbaum, P.~R. (2010a).
\newblock {\em Design of observational studies}, Volume~10.
\newblock Springer.

\bibitem[\protect\citeauthoryear{Rosenbaum}{Rosenbaum}{2010b}]{rosenbaum2010design}
Rosenbaum, P.~R. (2010b).
\newblock Design sensitivity and efficiency in observational studies.
\newblock {\em Journal of the American Statistical Association\/}~{\em
  105\/}(490), 692--702.

\bibitem[\protect\citeauthoryear{Rosenbaum}{Rosenbaum}{2011}]{rosenbaum2011aspects}
Rosenbaum, P.~R. (2011).
\newblock What aspects of the design of an observational study affect its
  sensitivity to bias from covariates that were not observed?
\newblock In {\em Looking Back: Proceedings of a Conference in Honor of Paul W.
  Holland}, pp.\  87--114. Springer.

\bibitem[\protect\citeauthoryear{Rosenbaum}{Rosenbaum}{2015}]{rosenbaum2015bahadur}
Rosenbaum, P.~R. (2015).
\newblock Bahadur efficiency of sensitivity analyses in observational studies.
\newblock {\em Journal of the American Statistical Association\/}~{\em
  110\/}(509), 205--217.

\bibitem[\protect\citeauthoryear{Rubin}{Rubin}{1980}]{rubin1980SUTVA}
Rubin, D.~B. (1980).
\newblock {Discussion of `Randomization analysis of experimental data: The
  Fisher randomization test comment' by Basu}.
\newblock {\em Journal of the American Statistical Association\/}~{\em
  75\/}(371), 591--593.

\bibitem[\protect\citeauthoryear{Rubin}{Rubin}{2007}]{rubin2007design}
Rubin, D.~B. (2007).
\newblock The design versus the analysis of observational studies for causal
  effects: parallels with the design of randomized trials.
\newblock {\em Statistics in medicine\/}~{\em 26\/}(1), 20--36.

\bibitem[\protect\citeauthoryear{Soriano, Ben-Michael, Bickel, Feller, and
  Pimentel}{Soriano et~al.}{2021}]{soriano2021interpretable}
Soriano, D., E.~Ben-Michael, P.~J. Bickel, A.~Feller, and S.~D. Pimentel
  (2021).
\newblock Interpretable sensitivity analysis for balancing weights.
\newblock {\em arXiv preprint arXiv:2102.13218\/}.

\bibitem[\protect\citeauthoryear{Tan}{Tan}{2006}]{tan2006distributional}
Tan, Z. (2006).
\newblock A distributional approach for causal inference using propensity
  scores.
\newblock {\em Journal of the American Statistical Association\/}~{\em
  101\/}(476), 1619--1637.

\bibitem[\protect\citeauthoryear{Tan}{Tan}{2007}]{tan2007comment}
Tan, Z. (2007).
\newblock Comment: Understanding or, ps and dr.
\newblock {\em Statistical Science\/}~{\em 22\/}(4), 560--568.

\bibitem[\protect\citeauthoryear{Van~der Vaart}{Van~der
  Vaart}{2000}]{van2000asymptotic}
Van~der Vaart, A.~W. (2000).
\newblock {\em Asymptotic statistics}, Volume~3.
\newblock Cambridge university press.

\bibitem[\protect\citeauthoryear{VanderWeele and Ding}{VanderWeele and
  Ding}{2017}]{vanderweele2017sensitivity}
VanderWeele, T.~J. and P.~Ding (2017).
\newblock Sensitivity analysis in observational research: introducing the
  e-value.
\newblock {\em Annals of internal medicine\/}~{\em 167\/}(4), 268--274.

\bibitem[\protect\citeauthoryear{Wang and Zubizarreta}{Wang and
  Zubizarreta}{2020}]{wang2020minimal}
Wang, Y. and J.~R. Zubizarreta (2020).
\newblock Minimal dispersion approximately balancing weights: asymptotic
  properties and practical considerations.
\newblock {\em Biometrika\/}~{\em 107\/}(1), 93--105.

\bibitem[\protect\citeauthoryear{Yang and Ding}{Yang and
  Ding}{2018}]{yang2018asymptotic}
Yang, S. and P.~Ding (2018).
\newblock Asymptotic inference of causal effects with observational studies
  trimmed by the estimated propensity scores.
\newblock {\em Biometrika\/}~{\em 105\/}(2), 487--493.

\bibitem[\protect\citeauthoryear{Zhang and Zhao}{Zhang and
  Zhao}{2022}]{zhang2022bounds}
Zhang, Y. and Q.~Zhao (2022).
\newblock Bounds and semiparametric inference in $l_\infty$ and
  $l_2$-sensitivity analysis for observational studies.
\newblock {\em arXiv preprint arXiv:2211.04697\/}.

\bibitem[\protect\citeauthoryear{Zhao, Small, and Bhattacharya}{Zhao
  et~al.}{2019}]{zhao2019sensitivity}
Zhao, Q., D.~S. Small, and B.~B. Bhattacharya (2019).
\newblock Sensitivity analysis for inverse probability weighting estimators via
  the percentile bootstrap.
\newblock {\em Journal of the Royal Statistical Society: Series B (Statistical
  Methodology)\/}.

\bibitem[\protect\citeauthoryear{Zubizarreta}{Zubizarreta}{2015}]{zubizarreta2015stable}
Zubizarreta, J.~R. (2015).
\newblock Stable weights that balance covariates for estimation with incomplete
  outcome data.
\newblock {\em Journal of the American Statistical Association\/}~{\em
  110\/}(511), 910--922.

\end{thebibliography}
\clearpage
%%%
%%% APPENDIX
%%%

\appendix
\setcounter{page}{1}
\singlespacing
\begin{center}
    \Large
    \textbf{Supplementary Materials: Design Sensitivity and Its Implications for Weighted Observational Studies}
\end{center}

\section{Proofs}
\label{sec:app_proofs}
\subsection{Theorem \ref{thm:power} (Power of a Sensitivity Analysis)} 
\label{app:pwr_sens}
For a general class of sensitivity models $\nu(\Gamma, w)$, define $\tau_{\nu(\Gamma, w)}$ as the minimum value in the set of possible point estimates (i.e., $\tau_{\nu(\Gamma, w)} := \inf_{\tilde w \in \nu(\Gamma, w)} \tau(\tilde w)$). Define $\xi_{\nu(\Gamma, w)} := \tau(w) - \tau_{\nu(\Gamma,w)}$. Finally, define $k_\alpha := 1- \Phi(\alpha)$. Then, the power of a sensitivity analysis is defined as: 

\begin{align}
\Pr\left(  \frac{\sqrt{n}(\hat \tau(\hat w) - \xi_{\nu(\Gamma, w)}) }{ \sigma_{\nu(\Gamma, w)}} \geq k_\alpha\right) %\nonumber \\
&= \Pr\left( \frac{\sqrt{n} \cdot ( \hat \tau(\hat w) - \tau(w))}{ \sigma_{W}} \geq \frac{k_\alpha \cdot \sigma_{\nu(\Gamma, w)} + \sqrt{n} \cdot (\xi_{\nu(\Gamma, w)}-\tau(w) )}{\sigma_{W}}\right) \nonumber \\
&\simeq 1 - \Phi \left( \frac{k_\alpha \cdot \sigma_{\nu(\Gamma, w)} + \sqrt{n} \cdot (\xi_{\nu(\Gamma, w)} - \tau(w))}{\sigma_{W}} \right),\nonumber  
\end{align} 
\begin{proof} 
\begin{align*} 
\Pr\left( \frac{\sqrt{n}(\hat \tau(\hat w) - \xi_{\nu(\Gamma, w)}}{ \sigma_{\nu(\Gamma, w)}} \geq k_\alpha\right)  
&= \Pr\left( \frac{\sqrt{n}(\hat \tau(\hat w) - \xi_{\nu(\Gamma, w)})}{\sigma(w)} \geq k_\alpha \cdot \frac{\sigma_{\nu(\Gamma, w)}}{\sigma(w)} \right) 
\intertext{Adding and subtracting $\sqrt{n} \tau(w)/\sigma(w)$ to both sides results in the following:}
&=  \Pr\left( \frac{\sqrt{n} \cdot ( \hat \tau(\hat w) - \tau(w))}{ \sigma_{W}} \geq \frac{k_\alpha \cdot \sigma_{\nu(\Gamma, w)} + \sqrt{n} \cdot (\xi_{\nu(\Gamma, w)}-\tau(w) )}{\sigma_{W}}\right) 
\intertext{Noting that $\hat \tau(\hat w) \cid N(\tau(w), \sigma(w)^2)$ concludes the proof:}
&\simeq 1 - \Phi \left( \frac{k_\alpha \cdot \sigma_{\nu(\Gamma, w)} + \sqrt{n} \cdot (\xi_{\nu(\Gamma, w)} - \tau(w))}{\sigma_{W}} \right)
\end{align*} 
\end{proof} 

\noindent \emph{Remark.}  Theorem \ref{thm:ds} relies on the asymptotic normality of the weighted estimator $\hat \tau(\hat w)$.  Several authors discuss assumptions sufficient to guarantee this.  \citet{zhao2019sensitivity} and \citet{huang2022variance} assume that the true propensity score obeys a logistic regression model in the observed covariates with parameter $\beta$ and is estimated via maximum likelihood. In this case the following conditions suffice, letting $\theta \in \Theta$ be a vector of parameters that includes $\beta$ as well as first and second moments for the outcome variables and the weights and the population weighted outcome mean.
\begin{assumption}[Regularity Conditions] \label{assump:regularity_conds}
Assume that the parameter space $\Theta$ is compact, and that $\theta$ is in the interior of $\Theta$. Furthermore, $(Y, X)$ satisfies the following: 
\begin{enumerate} 
\item $\E(Y^4) <\infty$
\item $\det \left( \E \left( \frac{\exp(\beta^\top X)}{(1+\exp(\beta^\top X))^2} X X^\top \right) \right) > 0$
\item $\forall$ compact subsets $S \subset \R^d$, $\E(\sup_{\beta \in S} \exp(\beta^\top X) ) < \infty$
\end{enumerate} 
\end{assumption} 
\citet[Assumption 4]{soriano2021interpretable}
discuss more general assumptions appropriate for balancing weights approaches that do not rely on a specific parametric model;  \citet[Theorem 3]{wang2020minimal} and \citet[Theorem 2]{hirshberg2019minimax} offer further versions of such assumptions.

\subsection{Theorem \ref{thm:design_sensitivity_vbm} (Design Sensitivity for the Variance-Based Sensitivity Model)} 
\label{subsec:ds_vbm}
\begin{proof} 
As an overview of the proof, we will first show that $\sigma_{\nu_{vbm}(R^2, w)} < \infty$. Then, we will invoke the results from Theorem \ref{thm:power} to solve for the $R^2$ parameter for which $\xi_{\nu_{vbm}(R^2, w)}$ is equal to $\tau(w)$. In order for $\sigma_{\nu_{vbm}(R^2,w)} < \infty$, the endpoints of the range of potential point estimates under the variance-based sensitivity model must have finite variance. This is a secondary result, proven in \cite{huang2022variance}, Theorem 3.2.

Because $\sigma_{\nu_{vbm}(R^2, w)} < \infty$, the results for the theorem follow almost immediately from Theorem \ref{thm:power}. Recall that the design sensitivity is defined as the minimum parameter value for a set of sensitivity models $\nu$ for which $\xi_{\nu(\Gamma, w)} > \tau(w)$ (i.e., when $\sqrt{n}(\xi_{\nu(\Gamma, w)} - \tau(w)) =0$). To solve for the design sensitivity for the variance-based sensitivity model, we begin by noting that the error term $\xi_{\nu_{vbm}(R^2, w)}$ is equal to the maximum bias for a set of sensitivity models. Following \cite{huang2022variance}, the maximum asymptotic bias that can occur is given by: 
\begin{align*} 
\xi&_{\nu_{vbm}(R^2, w)} \\
&:= \max_{\tilde w \in \nu_{vbm}(R^2)} \text{Bias}(\tau(w) \mid \tilde w) \\
&= \sqrt{1-\cor(w, Y \mid Z = 0)^2} \cdot \sqrt{\frac{R^2}{1-R^2} \cdot \var(w \mid Z = 0) \cdot \var(Y \mid Z = 0)}.
\end{align*} 

\noindent To solve for $\tilde R^2$, we set $\xi_{\nu_{vbm}(R^2, w)}$ equal to $\tau(w)$: 
\begin{align*} 
\sqrt{1-\cor(w, Y \mid Z = 0)^2} \cdot \sqrt{\frac{\tilde R^2}{1-\tilde R^2} \cdot \var(w \mid Z = 0) \cdot \var(Y \mid Z = 0)} = \tau(w) \\
\frac{\tilde R^2}{1-\tilde R^2} = \frac{1}{1-\cor(w, Y \mid Z = 0)^2} \cdot \frac{\tau^2_W}{\var(w \mid Z = 0) \mid \var(Y \mid Z = 0)}\\
\tilde R^2 = \frac{a^2}{1+a^2} \text{ where } a^2 = \frac{1}{1-\cor(w, Y \mid Z = 0)^2} \cdot \frac{\tau(w)^2}{\var(w \mid Z = 0) \cdot \var(Y \mid Z = 0)}
\end{align*} 
\end{proof} 

\subsection{Theorem \ref{thm:design_sensitivity_msm} (Design Sensitivity for the Marginal Sensitivity Model)}
\label{sec:proof_msm}

\begin{proof}
For a fixed value of $\Lambda$, we reject the null hypothesis that $\tau = 0$ if
\begin{align*}
    &\underset{n \to \infty}{\lim} \min_{\tilde w \in \nu_{msm}(\Lambda,w)} \hat \tau_{\tilde w} > 0 \\
    \iff &\underset{n \to \infty}{\lim} \left[\frac{1}{\sum_{i=1}^n Z_i} \sum_{i=1}^n Y_i Z_i - \max_{\tilde w \in \nu_{msm}(\Lambda,w)} \frac{\sum_{i=1}^n \tilde w_i Y_i (1-Z_i)}{\sum_{i=1}^n \tilde w_i (1-Z_i)} \right] > 0 \\
    \iff & \underset{n \to \infty}{\lim} \frac{1}{\sum_{i=1}^n Z_i} \sum_{i=1}^n Y_i Z_i  > 
    \underset{n \to \infty}{\lim} \max_{\tilde w \in \nu_{msm}(\Lambda,w)} \frac{\sum_{i=1}^n \tilde w_i Y_i (1-Z_i)}{\sum_{i=1}^n \tilde w_i (1-Z_i)}.
\end{align*}
Therefore, we can compute the design sensitivity $\Tilde{\Lambda}$ by finding $\Lambda$ such that
\begin{align}
    &\underset{n \to \infty}{\lim} \frac{1}{\sum_{i=1}^n Z_i} \sum_{i=1}^n Y_i Z_i  = 
    \underset{n \to \infty}{\lim} \max_{\tilde w \in \nu_{msm}(\Lambda,w)} \frac{\sum_{i=1}^n \tilde w_i Y_i (1-Z_i)}{\sum_{i=1}^n \tilde w_i (1-Z_i)}. \label{eq:msm_des_sens_est_eq}
\end{align}

The term on the left hand side of the estimating equation \eqref{eq:msm_des_sens_est_eq} is the observed data sample mean of $Y(1)$ and is equal to $\mathbb{E}\left[Y(1) \mid Z = 1\right]$ by the law of large numbers. We focus on showing that the right hand side limit exists and computing its value.

For notational simplicity, let $\hat{\mu}_0(\tilde w) := \frac{\sum_{i=1}^n \tilde w_i Y_i (1-Z_i)}{\sum_{i=1}^n \tilde w_i (1-Z_i)}$. Without loss of generality, let the first $m$ units be control units such that $Z_1 = \cdots = Z_m = 0, Z_{m+1} = \cdots = Z_n = 1$, where $1 \leq m < n$. Additionally, let $Y$ be ordered from largest to smallest such that $Y_1 \geq Y_2 \geq \cdots \geq Y_m$ and let $Y_i = 0$ for $i \notin \{1,\ldots,m\}$. Then, by Proposition 2 from \citet{zhao2019sensitivity},
\begin{align*}
    \max_{\tilde w \in \nu_{msm}(\Lambda,w)} \hat{\mu}_0(\tilde w) 
    &= \underset{a\in \{0,\ldots,m\}}{\max} \; \frac{\sum\limits_{i = \min\{a,1\}}^a\Lambda w_iY_i + \sum\limits_{i = \min\{a+1,m+1\}}^{\max\{a+1,m\}} \frac{1}{\Lambda} w_iY_i}{\sum\limits_{i = \min\{a,1\}}^a\Lambda w_i + \sum\limits_{i = \min\{a+1,m+1\}}^{\max\{a+1,m\}}\frac{1}{\Lambda} w_i} \\
    &= \underset{c\in \mathbb{R}}{\max} \; \frac{ \sum\limits_{i =1}^m \Lambda \mathbbm{1}\left\{Y_i \geq c \right\} w_iY_i + \sum\limits_{i =1}^m \frac{1}{\Lambda} \mathbbm{1}\left\{Y_i < c \right\} w_iY_i}{ \sum\limits_{i =1}^m \Lambda \mathbbm{1}\left\{Y_i \geq c \right\} w_i +  \sum\limits_{i =1}^m \frac{1}{\Lambda} \mathbbm{1}\left\{Y_i < c \right\} w_i}.
\end{align*}

The following lemma allows us to state the limit of $\underset{\tilde w \in \nu_{msm}(\Lambda,w)} {\max}\hat{\mu}_0(\tilde w)$.

\begin{lemma}[Limit of $\underset{\tilde w \in \nu_{msm}(\Lambda,w)} {\max}\hat{\mu}_0(\tilde w)$]\label{lemma:mu_0_lim}\mbox{}\\
Under Assumption \ref{assumption:overlap} and $\mathbb{E}\left(Y_i^2\right), \mathbb{E}\left(w_i^2\right) < \infty$,
\begin{align}
    \max_{\tilde w \in \nu_{msm}(\Lambda,w)} \hat{\mu}_0(\tilde w) = \;
    & \underset{c\in \mathbb{R}}{\max} \; \frac{ \sum\limits_{i =1}^m \Lambda \mathbbm{1}\left\{Y_i \geq c \right\} w_iY_i + \sum\limits_{i =1}^m \frac{1}{\Lambda} \mathbbm{1}\left\{Y_i < c \right\} w_iY_i}{ \sum\limits_{i =1}^m \Lambda \mathbbm{1}\left\{Y_i \geq c \right\} w_i +  \sum\limits_{i =1}^m \frac{1}{\Lambda} \mathbbm{1}\left\{Y_i < c \right\} w_i} \label{eq:max_lemma} \\
    \overset{p}{\to} \;& \underset{c\in \mathbb{R}}{\max} \; \frac{  \Lambda \mathbb{E}\left[\mathbbm{1}\left\{Y(0) \geq c \right\} w Y(0)|Z=0\right] + \frac{1}{\Lambda} \mathbb{E}\left[\mathbbm{1}\left\{Y(0) < c \right\} w Y(0)|Z=0\right]}{  \Lambda \mathbb{E}\left[\mathbbm{1}\left\{Y(0) \geq c \right\} w |Z=0\right] + \frac{1}{\Lambda} \mathbb{E}\left[\mathbbm{1}\left\{Y(0) < c \right\} w |Z=0\right]} \label{eq:lim_lemma}
\end{align}
\end{lemma} 
\begin{proof} 

We break $\underset{\tilde w \in \nu_{msm}(\Lambda,w)} {\max}\hat{\mu}_0(\tilde w)$ into four functions and show that each function converges uniformly to its corresponding expectation. As a result, $\underset{\tilde w \in \nu_{msm}(\Lambda,w)} {\max}\hat{\mu}_0(\tilde w)$ converges to its expectation. First, let

\begin{align*}
    \underset{\tilde w \in \nu_{msm}(\Lambda,w)} {\max}\hat{\mu}_0(\tilde w) = \underset{c\in \mathbb{R}}{\max} \; \frac{ \Lambda \overline{g}_1(c) + \frac{1}{\Lambda} \overline{g}_2(c) }{ \Lambda \overline{g}_3(c) + \frac{1}{\Lambda} \overline{g}_4(c) },
\end{align*}
where $\overline{g}_t(c) = \frac{1}{m}\sum\limits_{i =1}^m g_t(Y_i,w_i;c)$ for $t \in\{ 1, 2,3,4\}$ and
\begin{enumerate}
    \item ${g}_1(Y_i,w_i;c) = \mathbbm{1}\left\{Y_i \geq c \right\} w_iY_i$
    \item ${g}_2(Y_i,w_i;c) =  \mathbbm{1}\left\{Y_i < c \right\} w_iY_i$
    \item ${g}_3(Y_i,w_i;c) =  \mathbbm{1}\left\{Y_i \geq c \right\} w_i$
    \item ${g}_4(Y_i,w_i;c) =   \mathbbm{1}\left\{Y_i < c \right\} w_i$.
\end{enumerate}
We show that the class of functions $\mathcal{F} = \{g_1(y,w;c): c \in \mathbb{R} \cup \{-\infty, \infty\}\}$, each element of which maps $(Y_i,w_i)$ to the real line, is Glivenko-Cantelli and therefore $\overline{g}_1(c)$ converges uniformly to its expectation, $\mathbb{E}\left[\mathbbm{1}\left\{Y(0) \geq c \right\} w Y(0)|Z=0\right]$. A similar result can be shown for $\overline{g}_2(c)$, $\overline{g}_3(c)$, and $\overline{g}_4(c)$.  To clarify our exposition, we focus initially on the case in which the distribution of $Y$ is continuous.

Let $P$ be the probability distribution from which $(Y_1, w_1), ..., (Y_m, w_m)$ is a random sample and let $F$ be the cdf of $Y$. Choose any $\epsilon > 0$. By our assumptions and the Cauchy-Schwarz theorem, $\mathbb{E}|Y_iw_i| < \infty$; therefore, there exist constants $M^-_\epsilon$ and $M^+_\epsilon$ sufficiently large such that $\mathbb{E}\left[|Y_iw_i|\mathbbm{1}\{Y_i < -M^-_\epsilon\}\right]< \epsilon$ and $\mathbb{E}\left[|Y_iw_i|\mathbbm{1}\{Y_i > M^+_\epsilon\}\right]  < \epsilon$.  Define $p^-_\epsilon = F(-M_\epsilon^-)$ and $p^+_\epsilon = F(M_\epsilon^+)$, and let  $\Delta_\epsilon^- = F(0)-p^-_\epsilon$ and $\Delta_\epsilon^+ =p^+_\epsilon - F(0)$. 
Finally, choose any $k \in \mathbb{N}$.

Define functions $f^-_j = \mathbbm{1}\left\{Y_i \geq  F^{-1}\left(p^-_\epsilon + \frac{j\Delta^-_\epsilon}{k}\right) \right\} w_iY_i$ and $f^+_j = \mathbbm{1}\left\{Y_i \geq F^{-1}\left(p^+_\epsilon-\frac{j\Delta^+_\epsilon}{k}\right) \right\} w_iY_i$ for $j \in \left\{ 0, \ldots, k\right\}$.  If $Y_i$ has only nonnegative (nonpositive) support, the quantities $M^-_\epsilon, p^-_\epsilon, \Delta^-_\epsilon, f^-_i$ $(M^+_\epsilon,  p^+_\epsilon, \Delta^+_\epsilon, f^+_i)$ are unnecessary.   In addition, let $\underline{f} = g_1(Y_i,w_i;-\infty) = Y_iw_i$, $\overline{f} = g_1(Y_i,w_i;\infty) = 0$, and note that $f^-_{k} = f^+_{k} = \mathbbm{1}\{Y_i \geq 0 \}w_iY_i$. % = f_k_{k_\epsilon}$.  
For any two real-valued functions $\ell(Y_i,w_i), u(Y_i,w_i)$ such that $\ell(Y_i,w_i), u(Y_i,w_i)$ for all $(Y_i,w_i)$, define the bracket $[\ell, u] = \{f \in \mathcal{F}: \ell(Y_i, w_i) \leq f(Y_i, w_i) \leq  u(Y_i, w_i)\}$ as the set of all functions contained between them. 

Then the brackets
\begin{align*}
    &\left[\underline{f}, f^-_0\right], \left[f^-_0, f^-_1\right], \left[f^-_{1}, f^-_{2}\right], \ldots, \left[f^-_{k-1}, f^-_{k}\right] \text{ and}\\
   &\left[\overline{f}, f^+_0\right],\left[f^+_0, f^+_1\right], \left[f^+_{1}, f^+_{2}\right], \ldots, \left[f^+_{k-1}, f^+_{k}\right] 
\end{align*}
form a coverage of the function class $\mathcal{F}$ since every function in $\mathcal{F}$ belongs to at least one bracket.

Now that we have a set of $2(k+1)$ brackets $[\ell, u]$ that cover $\mathcal{F}$, we show that they are $\epsilon$-brackets in the sense that $P(u - \ell) < \epsilon$ for bracket where $Pf = \int f dP$.  Let $C = E(w_i^2) < \infty$. 
\begin{align}
P|f^-_0 - \underline{f}| = \mathbb{E}\left|\mathbbm{1}\left\{Y_i < M^-_\epsilon \right\}w_iY_i\right| &< \epsilon \quad \quad \text{and}\\
P|f^+_0 - \overline{f}| = \mathbb{E}\left|\mathbbm{1}\left\{Y_i > M^+_\epsilon \right\}w_iY_i\right| &< \epsilon
\end{align}
by our initial choices of $M^-_\epsilon$ and $M^+_\epsilon$.  
For $j = 1, \ldots, k$,
\begin{align}
    P|f^-_j - f^-_{j-1}| & \leq \mathbb{E}\left[\mathbbm{1}\left\{Y_i \in \left[F^{-1}\left(p^-_\epsilon + \frac{(j-1)\Delta^-_\epsilon}{k}\right), F^{-1}\left(p^-_\epsilon +\frac{j\Delta^-_\epsilon}{k}\right) \right) \right\} w_i M^-_\epsilon \right] \nonumber \\
    & \leq {C M^-_\epsilon}
    \cdot \text{Pr}\left(Y_i \in \left[F^{-1}\left(p^-_\epsilon + \frac{(j-1)\Delta^-_\epsilon}{k}\right), F^{-1}\left(p^-_\epsilon +\frac{j\Delta^-_\epsilon}{k}\right) \right) \right) \nonumber \\
    & = \frac{C M^-_\epsilon\Delta^-_\epsilon}{k} \leq \frac{C M^-_\epsilon}{k}. \label{eq:proof_fminus}
\end{align}

The second line follow from the Cauchy-Schwarz inequality.  Similarly, 
\begin{align}
       P|f^+_j - f^+_{j-1}| & \leq \mathbb{E}\left[\mathbbm{1}\left\{Y_i \in \left[F^{-1}\left(p^+_\epsilon - \frac{j\Delta^+_\epsilon}{k}\right), F^{-1}\left(p^+_\epsilon -\frac{(j-1)\Delta^+_\epsilon}{k}\right) \right) \right\} w_i M^+_\epsilon \right] \nonumber \\
    & \leq {C M^+_\epsilon}
    \cdot \text{Pr}\left(Y_i \in \left[F^{-1}\left(p^+_\epsilon - \frac{j\Delta^+_\epsilon}{k}\right), F^{-1}\left(p^+_\epsilon -\frac{(j-1)\Delta^+_\epsilon}{k}\right) \right) \right) \nonumber \\
    & = \frac{C M^+_\epsilon\Delta^+_\epsilon}{k} \leq \frac{C M^+_\epsilon}{k}. \label{eq:proof_fplus}
\end{align}
Since $k$ was chosen arbitrarily, we can select a value large enough such that $\frac{C M^-_\epsilon}{k}, \frac{C M^+_\epsilon}{k} < \epsilon$.   Therefore, by Theorem 19.4 from \citet{van2000asymptotic}, since the bracketing numbers are finite for every $\epsilon >0 $, the class of functions $\mathcal{F}$ is $P$-Glivenko-Cantelli. Since $\mathcal{F}$ is Glivenko-Cantelli, by Theorem 19.1 from \citet{van2000asymptotic}, uniform convergence holds.
 
If the distribution of $Y$ is not continuous, the above argument works until statements (\ref{eq:proof_fminus}) and (\ref{eq:proof_fplus}), which may not hold because the probability that $Y_i$ lies in a small region may still be large if a point probability mass is contained within it.  We can modify the argument to account for such point masses as follows.  Consider the set of points $\mathcal{Y}$ for which $Y_i$ has a point probability mass greater than or equal to $\epsilon/2$; this set must be finite in cardinality. Increase $M_\epsilon^+$ and $M_\epsilon^-$ so that $-M_\epsilon^- < y < M_\epsilon^+ \quad \text{for all $y \in \mathcal{Y}$}$.   Choose any $y_0 \in \mathcal{Y}$, and suppose without loss of generality that $y_0 < 0$. We can split any bracket $[f^-_j, f^-_{j+1}]$ that contains $g(Y_i,w_i; y_0)$ into the following two brackets:
\[
\left[f_j^-, \,\,  \mathbbm{1}\left\{Y_i \geq y_0 \right\} w_iY_i\right] \quad \quad \text{and} \quad \quad \left[\mathbbm{1}\left\{Y_i > y_0 \right\} w_iY_i,  \,\,f_{j+1}^- \right]
\]
Since the largest remaining probability point masses are all smaller than $\epsilon/2$, it is now possible to choose $k$ sufficiently large that each bracket $[\ell, u]$ satisfies $P|u - \ell| < \epsilon$.

If Lemma \ref{lemma:mu_0_lim} is not true, then assume that there is some value $\epsilon$ for which we can always find some $n$ such that \eqref{eq:max_lemma} and \eqref{eq:lim_lemma} are different by at least $\epsilon$. Since each of the functions $\overline{g}_1(c)$, $\overline{g}_2(c)$, $\overline{g}_3(c)$, and $\overline{g}_4(c)$ differ from their expectations by at most $\epsilon_1$, $\epsilon_2$, $\epsilon_3$, and $\epsilon_4$, respectively, we can construct a new $\epsilon$ which upper bounds the difference between \eqref{eq:max_lemma} and \eqref{eq:lim_lemma}. If there is not uniform convergence of \eqref{eq:max_lemma} and \eqref{eq:lim_lemma}, then the two terms have to be different by at least $\epsilon$. Therefore, Lemma \ref{lemma:mu_0_lim} follows by contradiction.
\end{proof}

Noting that \eqref{eq:lim_lemma} can equivalently be written with indicator functions in terms of the conditional CDF of $Y(0)$ given $Z = 0$, Theorem \ref{thm:design_sensitivity_msm} follows from Lemma \ref{lemma:mu_0_lim}.
\end{proof}

\subsubsection{Corollaries to Theorem \ref{thm:design_sensitivity_msm} for Trimmed and Augmented Weighting Estimators}
\label{sec:proof_msm_corollary}

\begin{corollary}[Design Sensitivity for the Marginal Sensitivity Model for Trimming]\label{corr:design_sensitivity_msm_trim}\mbox{}\\
Define $G_{\theta, m}(Y)$ as the following function: 
$$G_{\theta, m}(Y) = \begin{cases} 
1 & \text{if } Y \geq F^{-1}_{Y \mid w<m, Z=0}(1-\theta) \\
0 & \text{if } Y < F^{-1}_{Y \mid w<m, Z=0}(1-\theta)
\end{cases},$$
where $F_{y \mid x}$ represents the population c.d.f. of $y$ given $x$ under the favorable situation and $m$ represents the trimming cutoff. Let $\Tilde{\Lambda}$ be any solution to the following estimating equation:
\begin{align*} 
\mathbb{E}&[w Y(0) \mid w < m, Z = 0] + \tau_{trim} = \\
& \underset{\theta \in [0,1]}{\sup} \frac{\Lambda\mathbb{E}\left[w Y(0) \cdot G_{\theta, m}(Y(0)) \mid w < m, Z=0 \right]  + \frac{1}{\Lambda}\mathbb{E}\left[ w Y(0) \cdot (1-G_{\theta, m}(Y(0)))\mid w < m, Z=0\right]}{ \Lambda\mathbb{E}\left[w \cdot G_{\theta, m}(Y(0)) \mid w < m, Z=0 \right] + \frac{1}{\Lambda}\mathbb{E}\left[w \cdot (1-G_{\theta, m}(Y(0))) \mid w < m, Z=0 \right]},
\end{align*}

\noindent where $\tau_{trim} = \mathbb{E}\left[Y(1) - Y(0) \mid Z = 1, w < m\right]$. Then $\widetilde{\Lambda}$ is the design sensitivity.

\end{corollary}

\begin{proof}
    The proof of Corollary \ref{corr:design_sensitivity_msm_trim} is equivalent to the proof of Theorem \ref{thm:design_sensitivity_msm} after removing units with $w_i \geq m$.
\end{proof}

\begin{corollary}[Design Sensitivity for the Marginal Sensitivity Model for Augmentation]\label{corr:design_sensitivity_msm_aug}\mbox{}\\
Define $\Tilde{\Lambda}$ as any solution to the following estimating equation (where $F_{y \mid x}$ represents the population cdf of $y$ given $x$ under the favorable situation):
\begin{align*} 
\mathbb{E}&[w e \mid Z = 0] + \tau \\
&= \underset{\theta \in [0,1]}{\sup} \frac{\Lambda\mathbb{E}\left[w e \mathbbm{1}\left\{e \geq F^{-1}_{e \mid Z=0}(1-\theta)\right\} \mid Z=0 \right]  + \frac{1}{\Lambda}\mathbb{E}\left[ w e \mathbbm{1}\left\{e < F^{-1}_{e \mid Z=0}(1-\theta)\right\}\mid Z=0\right]}{ \Lambda\mathbb{E}\left[w \mathbbm{1}\left\{e \geq F^{-1}_{e \mid Z=0}(1-\theta)\right\} \mid Z=0 \right] + \frac{1}{\Lambda}\mathbb{E}\left[w \mathbbm{1}\left\{e < F^{-1}_{e \mid Z=0}(1-\theta)\right\} \mid Z=0 \right]},
\end{align*}
\noindent where $e := Y - g(X)$ are the residuals from an arbitrary outcome model $g$. Then $\widetilde{\Lambda}$ is the design sensitivity.

\end{corollary}

\begin{proof}
    The proof of Corollary \ref{corr:design_sensitivity_msm_aug} is equivalent to the proof of Theorem \ref{thm:design_sensitivity_msm} after replacing the outcomes with residuals.
\end{proof}

\subsection{Theorem \ref{thm:aug} (Impact of Augmentation on Design Sensitivity)} 
Define $e := Y - g(X)$ as the residual from an arbitrary outcome model $g$ used to augment a weighted estimate. Then, for the variance-based sensitivity model, the design sensitivity from an augmented weighted estimators will be greater than the design sensitivity for a standard weighted estimator if the following holds: 
$$\var(e \mid Z = 0) \leq \frac{1-\cor(w, Y \mid Z = 0 )^2}{1-\cor(w, e \mid Z = 0)^2} \cdot \var(Y \mid Z = 0)$$
\begin{proof} 
To begin, we will first derive the design sensitivity for the variance-based sensitivity model for augmented weighted estimators. If we treat the model $g(X)$ as fixed, then it is simple to show that under the same regularity assumptions as the ones invoked in Theorem \ref{thm:design_sensitivity_vbm} (i.e., Assumption \ref{assump:regularity_conds}), $\sigma^{aug}_{\nu(R^2, w)} < \infty$. In particular, we can apply the same proof, but substitute the residuals for the outcomes. The regularity conditions effectively state that the fourth moment of the residuals must be finite. 

Following \cite{huang2022sensitivity} (Theorem 5.1), note that the maximum asymptotic bias that can occur for an augmented weighted estimator is: 
\begin{align*} 
\xi^{aug}_{\nu_{vbm}(R^2, w)} :&= \max_{\tilde w \in \nu_{vbm}(R^2, w)} \text{Bias}(\tau(w)^{aug} \mid \tilde w) \\
&= \sqrt{1-\cor(w, e \mid Z = 0)^2} \cdot \sqrt{\frac{R^2}{1-R^2} \cdot \var(w \mid Z = 0) \cdot \var(e \mid Z = 0)}.
\end{align*} 
Then, because $\sigma^{aug}_{\nu(R^2, w)} < \infty$, following Theorem \ref{thm:design_sensitivity_vbm}, the design sensitivity can be algebraically solved for: 
$$\tilde R^2_{aug} = \frac{b^2}{1+b^2} \text{ where } b^2 = \frac{1}{1-\cor(w, e \mid Z = 0)^2} \cdot \frac{\tau_{aug}^{2}}{\var(w \mid Z = 0) \cdot \var(e \mid Z = 0)}$$
To compare $\tilde R^2_{aug}$ and $\tilde R^2$, we can re-write $\tilde R^2_{aug}$ as follows:
\begin{align*} 
\tilde R^2_{aug} &= \frac{b^2}{1+b^2} \\
&= \frac{\frac{1}{1-\cor(w, e \mid Z = 0)^2} \cdot \frac{\tau_{aug}^{2}}{\var(w \mid Z = 0) \cdot \var(e \mid Z = 0)}}{1+\frac{1}{1-\cor(w, e \mid Z = 0)^2} \cdot \frac{\tau_{aug}^{2}}{\var(w \mid Z = 0) \cdot \var(e \mid Z = 0)}}\\
&= \frac{\tau_{aug}^{2}}{(1-\cor(w, e \mid Z = 0))^2 \cdot \var(w \mid Z = 0) \cdot \var(e \mid Z = 0) + \tau_{aug}^2}
\end{align*} 
\noindent Then:
\begin{align*} 
\frac{\tilde R^2_{aug}}{\tilde R^2} &= \frac{(1-\cor(w, Y \mid Z = 0)^{2}) \cdot \var(w\mid Z = 0) \cdot \var(Y\mid Z = 0) + \tau(w)^2}{(1-\cor(w, e\mid Z = 0)^{2}) \cdot \var(w \mid Z = 0) \cdot \var(e \mid Z = 0) + \tau_{aug}^{2}} \cdot \frac{\tau_{aug}^{2}}{ \tau(w)^2}\\
&= \frac{ \tau^2_{aug}}{ \tau(w)^2} \cdot \frac{(1-\cor(w, Y \mid Z = 0)^2) \cdot \var(w) \cdot \var(Y) + \tau(w)^2}{(1-\cor(w, e\mid Z = 0)^2) \cdot \var(w) \cdot \var(e) + \tau^2_{aug}}
\end{align*} 
Because we are in the favorable setting, in which there is no omitted confounding, the weighted estimator will recover the estimand (i.e., the ATT) consistently (i.e., $\hat \tau(\hat w) \cip \tau(w) \equiv \tau$). Similarly, because the augmented weighted estimator is doubly robust, augmenting will also recover the estimand consistently (i.e., $\hat \tau_{aug} \cip \tau_{aug} \equiv \tau$), regardless of the outcome model. Thus, $\tau(w) = \tau_{aug}$. Then, the above is greater than 1 if the following criteria holds: 
\begin{align*} 
(1-\cor(w, e)^2) \cdot \var(e \mid Z = 0) &\leq (1-\cor(w, Y \mid Z = 0)^2) \cdot \var(Y \mid Z = 0) \\
\var(e \mid Z = 0) &\leq \frac{1-\cor(w, Y \mid Z = 0)^2}{1-\cor(w, e\mid Z = 0)^2} \cdot \var(Y \mid Z = 0) 
\end{align*} 
\end{proof} 

\subsection{Theorem \ref{thm:trim} (Impact of Trimming Weights on Design Sensitivity)}
Define some cutoff $m$ such that weights above the cutoff are trimmed. Furthermore, assume the trimmed weights are centered at mean 1 and the projection of the trimmed, ideal weights are centered on the trimmed, estimated weights. Then, for the variance-based sensitivity model, if the following holds: 
\begin{equation*} 
\underbrace{\frac{\var(w \mid w < m, Z = 0)}{\var(w\mid Z =0 )}}_{(1) \text{Variance reduction in } w} \leq \underbrace{\frac{1-\cor(w, Y \mid Z = 0)^2}{1-\cor(w, Y \mid w < m, Z = 0)^2}}_{(2) \text{ Change in relationship between } w \text{ and } Y} \cdot \underbrace{\frac{\var(Y \mid Z = 0)}{\var(Y \mid w < m, Z = 0)}}_{(3) \text{Variance reduction in } Y},
\end{equation*} 
the design sensitivity from a trimmed estimator will be greater than the standard weighted estimator. 
\begin{proof} 
Like Theorem \ref{thm:aug}, we will begin by deriving the design sensitivity for weighted estimators with trimming, under the variance-based sensitivity model. Furthermore, we have assumed that the trimmed weights are centered at mean 1: $\E(w \mid w < m) = 1$. This assumption  holds by construction whenever researchers normalize the trimmed weights appropriately; for more discussion of such normalization and its implications see the remark that follows this proof.

We will begin by deriving the maximum asymptotic bias for a trimmed weighted estimator. To begin, define the cutoff for trimming to be some threshold $m$, such that any observations $w \geq m$ are trimmed. Then, the estimand of interest is thus the average treatment effect, across the treated, subset to units that associated with weights $w < m$: 
\begin{align*} 
\tau^{trim} :&= \E(Y(1) - Y(0) \mid Z = 1, X \in \mathcal{A})\\
&\equiv \E(Y(1) - Y(0) \mid Z = 1, w < m)
\end{align*} 
Notably, the additional condition of $w < m$ is an observable condition, given the observed covariates $X$ (as the estimated weights are a function of $X$). The absolute population bias for a trimmed weighted estimator is: 
\begin{align*} 
&\Big| \tau_{trim}(w) - \tau_{trim} \Big|\\
&\text{where $\tau_{trim}(w) =  E(Y|Z=1, w < m)-\E(wY |Z=0, w < m)$. By conditional ignorability:} \nonumber \\
&=\left| \E(w Y \mid Z = 0, w < m) - \E(w^* Y \mid Z = 0, w < m) \right|\nonumber \\
&=\left| \E((w - w^*) \cdot Y \mid Z = 0, w < m) \right|\nonumber \\
&\text{By construction, $\E(w \mid Z = 0, w < m) = \E(w^* \mid Z = 0, w < m)$:} \nonumber \\
&= \left|\E((w - w^*) \cdot Y \mid Z = 0, w < m) - \E(w - w^* \mid Z = 0, w < m) \cdot \E(Y \mid Z = 0, w < m) \right|\nonumber \\
&=\left|\cov(w - w^*, Y \mid Z = 0, w < m) \right|\nonumber \\
&= \left|\cor(w - w^*, Y \mid Z = 0, w < m)\right| \cdot \sqrt{\var(w - w^* \mid Z = 0, w < m) \cdot \var(Y \mid Z = 0, w < m)}\\
&= \left|\cor(w - w^*, Y \mid Z = 0, w < m)\right| \cdot \sqrt{\var(w \mid Z = 0, w < m) \cdot \frac{R^2}{1-R^2} \var(Y \mid Z = 0, w < m)}
\end{align*} 
The last line follows from the fact the projection of the trimmed, ideal weights in the observed covariate space of $X$ are centered on the trimmed, estimated weights. For intuition, first define an indicator $V := \1\{ w > m\}$. Then, the trimmed, ideal weights can be written as $w^* \cdot V$. Similarly, the trimmed, estimated weights can be written as $w \cdot V$. Then, if the projection of the ideal weights in $X$ are centered on $w$ (a condition that trivially is met when using inverse propensity score weights), it follows immediately that $\E(w^* \cdot V \mid X) = w \cdot V$. As a result, within the space of $w < m$, the residual error (i.e., $w^* - w$) is orthogonal to the estimated weights $w$.

To bound the correlation term, we apply the recursive formula of partial correlation:
\begin{align*} 
-\sqrt{1-\cor(w, Y \mid Z = 0, w < m)^2} &\leq \cor(w - w^*, Y \mid Z = 0, w < m) \\
&\leq \sqrt{1-\cor(w, Y \mid Z = 0, w < m)^2}
\end{align*} 

Then, the maximum asymptotic bias for a trimmed weighted estimator is: 
\begin{align} 
\xi^{trim}_{\sigma(R^2, w)} :=& \max_{\tilde w \in \sigma(R^2)} \text{Bias}(\hat \tau(\hat w)^{trim}) \nonumber \\
=& \sqrt{1-\cor(w, Y \mid Z = 0, w < m)^2} \cdot \nonumber \\
&~~\sqrt{\frac{R^2}{1-R^2} \cdot \var(w \mid Z = 0, w < m) \cdot \var(Y \mid Z = 0, w < m)}.
\label{eqn:trim_bias_bound} 
\end{align} 
\noindent Solving for the $R^2$ such that $\xi^{trim}_{\sigma(R^2, w)} = \tau(w)$, 
$$\tilde R^2_{trim} = \frac{c^2}{1+c^2},$$
where 
$$c^2 := \frac{1}{1-\cor(w, Y \mid Z = 0, w < m)^2} \cdot \frac{\tau(w)^{2trim}}{\var(w \mid w < m, Z = 0) \cdot \var(Y \mid w < m, Z = 0)}.$$
Then, 
\begin{align*} 
&\frac{\tilde R^2_{trim}}{\tilde R^2} =\\
&\frac{(1-\cor(w, Y \mid Z = 0)^2) \cdot \var(w\mid Z = 0) \cdot \var(Y \mid Z = 0) + \tau(w)^2}{(1-\cor(w, Y \mid Z = 0, w < m)^2) \cdot \var(w\mid w < m,Z=0) \cdot \var(Y \mid w < m, Z = 0) + \tau(w)^{2trim}} %\\
%&\hspace{3em}
\cdot \frac{\tau(w)^{2trim}}{\tau(w)^2}\\
&= \frac{\tau(w)^{2trim}}{\tau(w)^2} \frac{(1-\cor(w, Y\mid Z = 0)^2) \cdot \var(w\mid Z = 0) \cdot \var(Y\mid Z = 0) + \tau(w)^2}{(1-\cor(w, Y \mid Z = 0, w < m)^2) \cdot \var(w \mid w < m, Z = 0) \cdot \var(Y \mid w < m, Z = 0) + \tau(w)^{2trim}}
\intertext{Let $\tau(w)^{trim} := \tau(w) + c$:}
&= \frac{(\tau(w) + c)^{2}}{\tau(w)^2} \cdot \\
&\frac{(1-\cor(w, Y \mid Z = 0)^2) \cdot \var(w \mid Z = 0) \cdot \var(Y \mid Z = 0) +\tau(w)^2}{(1-\cor(w, Y \mid Z = 0, w < m)^2) \var(w\mid Z =0, w < m) \var(Y \mid Z = 0, w < m) + (\tau(w)+c)^{2}}.
\end{align*} 
In order for there to be an improvement in design sensitivity from trimming, the following must hold: 
\begin{align*} 
\frac{(\tau(w) + c)^{2}}{\tau(w)^2} &\cdot \left((1-\cor(w, Y \mid Z = 0)^2) \cdot \var(w \mid Z = 0) \cdot \var(Y \mid Z = 0) +\tau(w)^2 \right) \geq \\
(1-\cor&(w, Y \mid Z = 0, w < m)^2) \cdot \var(w \mid Z = 0, w < m) \cdot \var(Y \mid w < m, Z = 0)\\ 
+ (\tau(w)&+c)^{2} 
\end{align*} 
Re-arranging:
$$\frac{\var(w \mid Z = 0)}{\var(w \mid Z = 0, w < m)} \geq \frac{1-\cor(w, Y \mid Z = 0, w < m)}{1-\cor(w, Y \mid Z = 0)^2} \cdot \left( \frac{\tau(w)}{\tau(w) + c}\right)^2 \cdot \frac{\var(Y \mid Z = 0, w < m)}{\var(Y\mid Z = 0)}$$

\noindent Because we are assuming a constant treatment effect, $c=0$, which allows us to arrive at Equation \eqref{eqn:trimming_criteria}: 
$$\frac{\var(w \mid Z = 0)}{\var(w\mid Z=0, w < m)} \geq \frac{1-\cor(w, Y \mid Z = 0, w < m)}{1-\cor(w, Y \mid Z = 0)^2} \cdot \frac{\var(Y \mid Z = 0, w < m)}{\var(Y\mid Z = 0)}$$

The results of Theorem \ref{thm:trim} may be easily extended in settings when researchers are interested in a smoothed trimmed estimator. Due to the non-smoothness of trimming, traditional trimming methods ignore the uncertainty in the design stage from estimating weights and conduct inference excluding units with extreme estimated weights. \citet{yang2018asymptotic} develop a smooth trimming estimator that weights all units continuously, assigning extremely small weights to units with large weights instead of removing them, and is asymptotically linear. Therefore, the bootstrap can be used to construct confidence intervals. Design sensitivity considers the large sample limits of the weights, so design stage uncertainty is not present. Furthermore, in asymptotic settings, the smoothed and non-smoothed trimmed estimators are equivalent. 

However, in settings when researchers are interested in calculating the power of a sensitivity analysis, in addition to design sensitivity, it can be helpful to consider the smoothed trimmed estimator. In particular, the (standard) trimmed estimator is non-smooth, and as a result, will not be amenable to a bootstrap-style procedure to estimating power \citep{yang2018asymptotic}. Following \cite{yang2018asymptotic}, we define a smoothed trimmed estimator, denoted as $\hat \tau(\hat w)^{smooth}$, which approximates the trimmed estimator arbitrarily well using a tuning parameter $\epsilon > 0$. Applying Theorem \ref{thm:power}, the power of a sensitivity analysis for $\hat \tau(\hat w)^{smooth}$ is as follows: 
\begin{align} 
\Pr \Big(& \frac{\sqrt{n}(\hat \tau(\hat w)^{smooth} - \xi_{\nu(\Gamma, w)}^{smooth}}{\sigma_{\nu(\Gamma, w)}^{smooth}} \geq k_\alpha \Big) \nonumber \\
&=  \Pr\left( \frac{\sqrt{n} \cdot ( \hat \tau^{smooth}(w) - \tau^{trim})}{ \sigma^{smooth}_{W}} \geq \frac{k_\alpha \cdot \sigma^{smooth}_{\nu(\Gamma, w)} + \sqrt{n} \cdot (\xi^{smooth}_{\nu(\Gamma, w)}-\tau^{trim} )}{\sigma^{smooth}_{W}}\right) \nonumber \\
&= \Pr\left( \frac{\sqrt{n} \cdot ( \hat \tau^{smooth}(w) - \tau^{trim})}{ \sigma^{smooth}_{W}} \geq \frac{k_\alpha \cdot \sigma^{smooth}_{\nu(\Gamma, w)} + \sqrt{n} \cdot (\xi^{trim}_{\nu(\Gamma, w)} + \delta -\tau^{trim} )}{\sigma^{smooth}_{W}}\right) \label{eqn:delta_smooth} \\
&\to 1 - \Phi \left( \frac{k_\alpha \cdot \sigma^{smooth}_{\nu(\Gamma, w)} + \sqrt{n} \cdot (\xi^{trim}_{\nu(\Gamma, w)} + \delta - \tau^{trim})}{\sigma^{smooth}_{W}} \right),\nonumber  
\end{align} 
where $\delta$ is a function of $\epsilon$. The expression from Equation \eqref{eqn:delta_smooth}, which introduces the $\delta$ constant follows directly from the fact that we may express the bias for the smoothed trimming estimator as a function of the bias of a standard (non-smooth) trimmed estimator and a function of $\epsilon$: 
\begin{align*} 
\text{Bias}(\hat \tau(w)^{smooth}) &= \E(\hat \tau(w)^{smooth}) - \tau^{trim} \\
&= \underbrace{ \E(\hat \tau(w)^{smooth}) - \E(\hat \tau(w)^{trim})}_{(*)} + \underbrace{\E(\hat \tau(w)^{trim}) -  \tau^{trim}}_{\equiv \text{Bias}(\hat \tau(w)^{trim})} 
\end{align*} 
The first term (denoted by $(*)$) is equal to $\delta$, which can be made arbitrarily large or small by tuning $\epsilon > 0$:
\begin{align*} 
 \E(&\hat \tau(w)^{smooth}) - \E(\hat \tau(w)^{trim}) \\
 =& \E \left(\frac{1}{\sum_{i=1}^n (1-Z) w V'} \sum_{i=1}^n (1-Z)w \cdot V'
Y  \right) - \\
&\E \left( \frac{1}{\sum_{i=1}^n (1-Z) w \1\{w < m\}} \sum_{i=1}^n (1-Z)w \cdot \1\{w < m\} Y \right) \\
=& \E(w (V' - \1\{w < m\}) Y \mid Z = 0) \\ 
=& \underbrace{\E(w (V' - \1\{w < m\}) Y \mid Z = 0,w < m) P(w < m \mid Z = 0)}_{= 0} + \\
&\E(w (V' - \1\{w < m\}) Y \mid Z = 0,w \geq m) P(w \geq m \mid Z = 0)  \\ 
:=& \varepsilon \cdot \E(w Y \mid Z = 0,w \geq m) P(w \geq m \mid Z = 0)  \\
\equiv& \delta 
 \end{align*} 

As such, the bias bound of the smoothed trimmed estimator can be written as the bias bound in Equation \eqref{eqn:trim_bias_bound} and an arbitrary constant $\delta$, which is a function of $\epsilon$: 
\begin{align*} 
\xi^{smooth}_{\sigma(R^2, w)} &= \max_{\tilde w \in \nu(\Gamma, w)} \text{Bias}(\hat \tau(w)^{smooth})\\
&= \max_{\tilde w \in \nu(\Gamma, w)} \text{Bias}(\hat \tau(w)^{trim}) + \delta \\ 
&= \xi_{\nu(\Gamma, w)}^{trim} + \delta
\end{align*} 
 
Therefore, we have shown that as $n \to \infty$, for an arbitrarily small $\epsilon > 0$, the design sensitivity will be within a $\delta$-neighborhood of the value $\tilde \Gamma$, for which $\xi^{trim}_{\nu(\Gamma, w)} = \tau(w)$.  
\end{proof}

\textbf{Remark:} Consider computing weights $w^{raw}$, trimming these weights at an upper threshold $m$, and then renormalizing the remaining weights by their sum, producing a final set of weights $w$.  Renormalized weights tend to have many attractive properties \citep{hajek1971comment}, and Theorem \ref{thm:trim} assumes such weights are used after trimming.  The trimmed estimand, defined originally by the condition $w^{raw} < m$, can be described in terms of the new weights $w$ as follows:
\[
E[Y(1) - Y(0)\mid w^{raw} < m] = E\left[Y(1) - Y(0)\left| w < \frac{m}{\mu_w}\right.\right],
\]
where $\mu_w = E(w^{raw})$ is the population normalization term.  As such, we can describe any analysis based on trimming raw weights equivalently as an analysis based on trimming on normalized post-trimming weights, although with a slightly different trimming threshold.

\section{Using a Planning Sample to Estimate Design Sensitivity}\label{app:planning_sample} 
\subsection{Calibrating Design Sensitivity to Outcome Data with a Planning Sample}
To estimate design sensitivity, researchers must posit an outcome model. One way to calibrate their priors to the existing outcome data is to utilize a planning sample. This is done by holding out part of the sample to use as a `planning sample' (akin to a pilot sample in experimental studies). The remainder of the sample is then used for the analysis. We will refer to the holdout sample as the `analysis sample.' We will assume in this section that researchers are randomly sampling observations from a fixed dataset to construct a planning sample. However, in cases when researchers have access to auxiliary outcome data (i.e., historical datasets), they may utilize these external datasets as the planning sample, and treat the full observational study as the analysis sample.

We will outline two approaches that researchers may use to estimate design sensitivity using a planning sample. The first approach proposes drawing a planning sample, and simply estimating the design sensitivity across the planning sample. Table \ref{tbl:plan_sample} outlines in more detail. 

\begin{table}[!ht]
\centering
\caption{\underline{\textbf{Estimating Design Sensitivity using a Planning Sample}}} \label{tbl:plan_sample}
\noindent\fbox{%
\vspace{2mm}
\parbox{0.95\textwidth}{%
\vspace{2mm}
\begin{Step} 
\item Fix an effect size $\tau$.
\item Estimate the weights $\hat w_i$ using the full sample.  
\item Across the units in the control group, generate a planning sample by randomly sampling $n_{plan}$ observations. Denote the set of indices that correspond to the planning sample as $\mathcal{P}$.
\item  If using the variance-based sensitivity model: 
\begin{itemize} 
\item[a.] Calculate the sample variance of the outcomes (i.e., $\widehat{\var}(Y_i \mid i \in \mathcal{P})$) and the sample correlation between the outcomes and the estimated weights (i.e., $\widehat{\cor}(\hat w_i, Y_i \mid i \in \mathcal{P})$). 
\item[b.] Calculate the sample variance of the estimated weights across the full sample.
\end{itemize} 
If using the marginal sensitivity model:
\begin{itemize} 
\item[a.] Calculate the weighted average of the control outcomes in the planning sample:
$$\hat \mu_0^{plan} \leftarrow \frac{\sum_{i:i \in \mathcal{P}} \hat{w}_i Y_i}{\sum_{i:i \in \mathcal{P}} \hat{w}_i}.$$
\item[b.] Generate the average treatment outcome, given a fixed $\tau$: 
$$\hat \mu_1 \leftarrow \tau +  \frac{\sum_{i:i \in \mathcal{P}} \hat{w}_i Y_i}{\sum_{i:i \in \mathcal{P}} \hat{w}_i}.$$
\end{itemize} 
\item Using the components generated in Step 4, estimate the design sensitivities under the variance-based sensitivity model using Theorem \ref{thm:design_sensitivity_vbm} and the marginal sensitivity model using Theorem \ref{thm:design_sensitivity_msm}.
\end{Step} 
}
}
\end{table} 

The approach outlined in Table \ref{tbl:plan_sample} allows researchers to calibrate the quantities needed to calculate design sensitivity using a planning sample. However, in settings where the outcome distribution may be heavy tailed, the planning sample may be unable to capture the full complexity present in the outcomes, which can result in an over-estimation of design sensitivity. In particular, this is of concern to the marginal sensitivity models, in which robustness to unmeasured confounding and the design sensitivity are often characterized by outliers \citep{huang2022variance}. One alternative way to leverage a planning sample, but additionally account for more complex variation across the full dataset is by first fitting an outcome model across the planning sample units, and use this model to simulate outcomes across the units in the analysis sample. Table \ref{tbl:plan_sample_sim} summarizes the procedure.

\begin{table} 
\centering 
\caption{\underline{\textbf{Estimating Design Sensitivity using a Planning Sample and Simulated Outcomes}}} \label{tbl:plan_sample_sim}
\noindent\fbox{%
\vspace{2mm}
\parbox{0.95\textwidth}{%
\vspace{2mm}
\begin{Step} 
\item Fix an effect size $\tau$. 
\item Estimate the weights $\hat w_i$ using the full sample. 
\item Across the units in the control group, generate a planning sample by randomly sampling $n_{plan}$ observations. Denote the set of indices that correspond to the planning sample as $\mathcal{P}$.
\item Across the planning sample $\mathcal{P}$, fit an outcome model $\hat g_\mathcal{P}(X_i)$ for the control units. Use the residuals from the fitted outcome model to estimate the variance in the residuals (i.e., unexplained variation in the outcomes, denoted as $\hat \sigma^2_{e, \text{plan}}$). 
\item Simulate new data for the units in the analysis sample (i.e., $i \not \in \mathcal{P}$) by parametrically sampling residuals from $\hat g_\mathcal{P}(X_i)$: 
\begin{enumerate} 
        \item[a.] Randomly sample units across the analysis sample, with replacement. We refer to this as the \textit{bootstrap sample}, $\mathcal{B}$. 
        \item[b.] For all units in $\mathcal{B}$, estimate the outcome $Y^*_i(0)$: 
        $$Y^*(0) \leftarrow \hat g_\mathcal{P}(X_i) + \epsilon^*,$$
        where $\epsilon^*_i \sim N(0, \hat{\sigma}^2_{e,\text{plan}})$.
\end{enumerate}
\item Apply Steps 4 and 5 from Table \ref{tbl:plan_sample}, but using $\mathcal{B}$ instead of the planning sample $\mathcal{P}$.
\end{Step} 
}
}
\end{table} 

The procedure outlined in Table \ref{tbl:plan_sample_sim} allows researchers to flexibly calibrate design sensitivity using a planning sample. The fitted outcome model can be of arbitrary specification, and researchers can leverage flexible, black box machine learning models to estimate $Y_i(0)$. To simulate the noise $\epsilon^*_i$ in Step 5-(b), we currently assume the errors are normally distributed. However, researchers may relax this assumption and posit any, arbitrary distribution for the residuals, using the residuals across the planning sample to calibrate the necessary parameters for the distribution.

To illustrate the proposed procedure, we turn to the empirical application. For simplicity, we will focus on the setting in which researchers are interested in the impact of presidential support on support for the FARC peace deal. We draw 100 different planning samples from the data, and use Table \ref{tbl:plan_sample_sim} to estimate the design sensitivity for a variety of different effect sizes across both the variance-based and marginal sensitivity models. Figure \ref{fig:pilot_sample} provides a visualization for the distribution of estimated design sensitivities across the different planning samples. The estimated design sensitivities using the planning sample are mostly centered around the oracle design sensitivities, calibrated using the full dataset. \\

\noindent \textbf{Remark on Sample Boundedness.} From Figure \ref{fig:pilot_sample}, we see that as the effect size increases, the spread of estimated design sensitivities under the marginal sensitivity model from the planning samples also increases. This is likely because the marginal sensitivity model is susceptible to sample boundedness. More specifically, \cite{huang2022variance} highlighted that because of the inherent stabilization within the model, the marginal sensitivity model can only recover a worst-case bias bound defined by the range of observed control outcomes. As a result, as the effect size increases towards the sample bounds, $\Lambda \to \infty$. Because the range of observed control outcomes depends on the observed data, variation in the drawn planning sample can drive variation in the estimated design sensitivity for the marginal sensitivity model. The variance-based sensitivity model is not susceptible to sample boundedness; as a result, the estimated design sensitivities across different planning samples remain relatively stable, even as the effect size increases. \\

Both approaches proposed in the following subsection provide researchers with a way to estimate design sensitivity using a planning sample. We have also illustrated that the design sensitivities estimated from a planning sample are similar to the design sensitivities estimated using the full dataset. However, what is arguably most important in practice is that the \textit{relative} estimates of design sensitivity from different design choices, such as using trimmed weights or an augmented weighted estimator, stably inform the optimal design choice. See Section \ref{sec:power_farc} for more details.

\begin{figure}[!t]
\centering 
\textbf{Distribution of Design Sensitivities, with 100 Planning Samples} \\ \vspace{2mm} 
\includegraphics[width=0.49\textwidth]{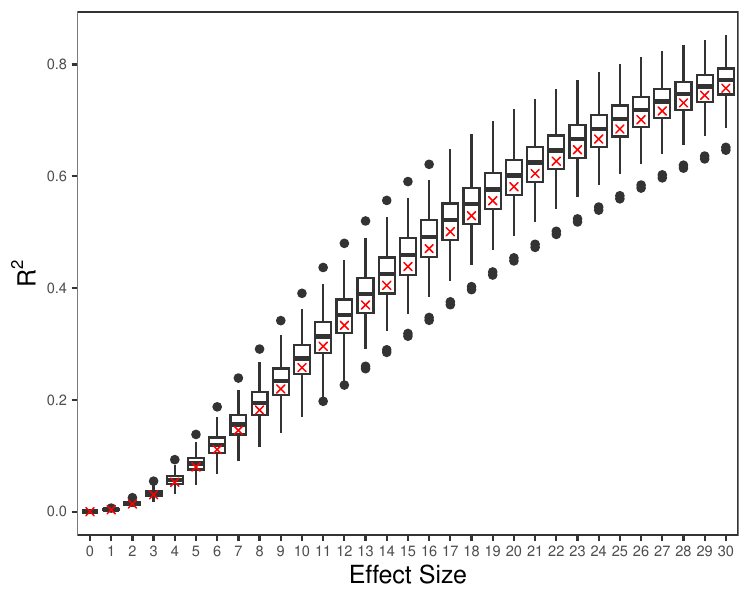}
\includegraphics[width=0.49\textwidth]{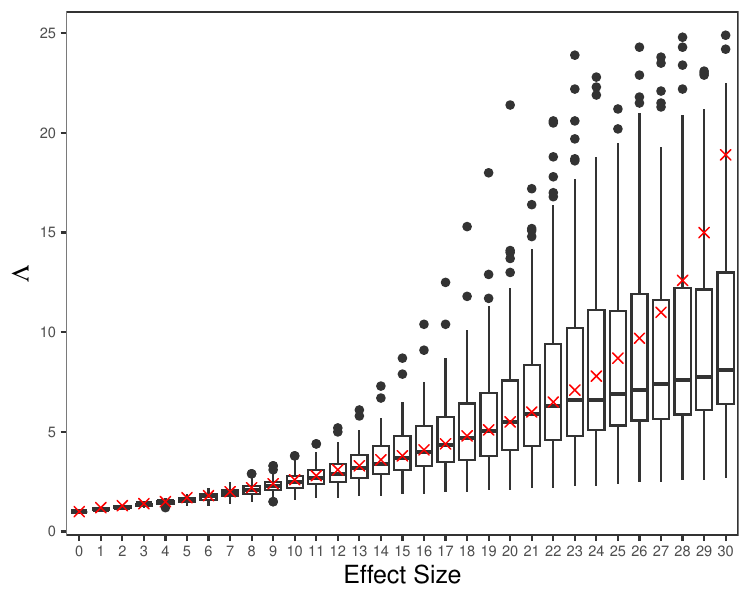}
\caption{Distribution of estimated design sensitivity measures, across 100 different planning samples. Design sensitivities were estimated using the procedure outlined in Table \ref{tbl:plan_sample_sim}. The red $\times$ points denote the oracle design sensitivities, calibrated using the full dataset.}
\label{fig:pilot_sample} 
\end{figure}

\subsection{Estimating Design Sensitivity for Augmentation, without Prior Specification of Outcome Model} 
We consider two settings for which researchers may estimate design sensitivity for augmentation. First, we consider the setting in which researchers are assessing whether or not they should fit an outcome model to begin with to perform augmentation. To estimate the impact of augmentation on design sensitivity without a prior specification of an outcome model, researchers can instead construct a \textit{proxy outcome model} that explains a fixed $r^2$ of the variation in the outcomes. They can then estimate the design sensitivity. If they find that an extremely high $r^2$ value is needed---i.e., they must explain a large percentage of the variation in the outcomes---for improvements in design sensitivity, this can be infeasible to do in practice, and choose to not augment their estimation.  Table \ref{tbl:aug_proxy} summarizes the procedure.

To make this more concrete, we can consider the empirical application. Assume first that researchers are interested in optimizing for robustness with respect to an average error (i.e., the variance-based sensitivity model). From Figure \ref{fig:farc}, we see that even if researchers were to estimate an outcome model that could explain 50\% of the variation, the design sensitivity for the variance-based sensitivity model would only improve slightly. In contrast, if researchers are worried about potential, worst-case confounding (i.e., the marginal sensitivity model), we see that if they were to augment the weighted estimator with an outcome model that could explain even only 25\% of the variation in the outcomes, the design sensitivity would improve substantially, especially in cases where the treatment effect might be relatively small.

In the second setting, researchers already have estimated an outcome model of interest \textit{a priori}. Then, design sensitivity can be estimated in the same manner as in the standard case. However, instead of calibrating to the outcome distribution, researchers must calibrate design sensitivity to the underlying residuals. 

It is worth noting that design sensitivity is usually not hurt from augmenting with an outcome model. However, employing a proxy outcome model first can help researchers determine if it is worth fitting a complex outcome model, and also better assess practical trade-offs, like whether to gather more covariate data that could feasibly help explain variation in the outcomes.

\begin{table}[!ht]
\caption{\textbf{\underline{Design Sensitivity for Augmentation with a Proxy Outcome Model}}} \label{tbl:aug_proxy} 
\centering
\noindent\fbox{%
\vspace{2mm}
\parbox{0.95\textwidth}{%
\vspace{2mm}
\begin{Step}
\item Set some $r^2$, which represents the variation explained in the outcome model. 
\item Draw a planning sample. 
\item Estimate a model $\hat m$ across the planning sample to generate outcomes.
\item Generate an outcome across the rest of the sample not in the planning sample. Add noise to each prediction to guarantee that the outcomes are the same variance as the outcomes in the planning sample. Denote this as $\tilde Y$. 
\item Construct a proxy outcome model $ g^*$ that can explain $r^2$ of the outcomes.
\begin{enumerate} 
\item Generate a scaled version of the outcomes: 
$$Z^* := (Y^* - \E(Y^*))/sd(Y^*)$$
\item Create a proxy covariate $X^*$ that is correlated with the generated outcomes: 
$$X^* := \sqrt{r^2} \cdot Z^* + \sqrt{1-r^2} \cdot v,$$
where $v$ is a standard normal random variable. 
\item Estimate a linear model with the proxy covariate on the outcomes
$$Y^*\sim X^*$$
\item From the linear model fit in the previous step, use the fitted values $\hat{{Y}^*}$ and the residuals $e^*$.
\end{enumerate} 
\item Estimate the design sensitivity. 
\end{Step}
}
}
\end{table} 

\clearpage 
\section{Detailed Simulation Results} \label{app:sim}

\subsection{Simulation Parameters} 
The simulation setup is under a \textit{favorable situation} defined as:
\begin{enumerate}
    \item The study is free of unmeasured bias. In the case of the marginal sensitivity model, the marginal sensitivity model is satisfied with $\Lambda = 1$. In the case of the variance-based sensitivity model, $R^2 = 0$. Equivalently, for both sensitivity models $w_i^* = w_i$.
    \item The null hypothesis of no treatment effect is false, and a specific alternative is true. In our case, the alternative we consider is that the data is generated by a stochastic model with a treatment effect as follows:
    \begin{align}
      P(Z_i = 1 \mid X_i) \propto \frac{\exp(\beta_\pi X_i)}{1+\exp(\beta_\pi X_i)}, \ \ \ \ \ \ Y_i = \beta_y X_i + \tau Z_i + u_i, \label{eq:sim_dgp}
    \end{align}
where $X_i \iid N(\mu_x, \sigma^2_x)$, and $u_i \iid N(0, \sigma^2_y)$.

\end{enumerate}
 We vary the different parameters, $\{\beta_\pi, \beta_y, \tau, \sigma_y \}$. We vary $\tau$ to control the effect size, $\beta_\pi$ for the variance in the weights, and $\beta_y$ and $\sigma_y$ to alter the variance in the outcomes. The correlation between the weights and the outcome depends on $\beta_\pi$, $\beta_y$, and $\sigma_y$. The base parameters are set as follows: $\tau = 1$, $\mu_x = 0$, $\sigma_x = 1$, $\beta_y = 1$, $\sigma_y = 1$, $\beta_\pi = 1$. For now, we will assume a constant treatment effect $\tau$ for all units $i$. However, we relax this assumption in Appendix \ref{subsec:het} and allow for heterogenous treatment effects. 

\subsection{Drivers of Design Sensitivity} 
To simulate the drivers of design sensitivity for each sensitivity model, we modify $\{\beta_\pi, \beta_y, \tau, \sigma_y\}$ such that one element of the data generating process changes, while holding the others constant. We are essentially estimating the derivative of the design sensitivity, with respect to the effect size, the variance in the estimated weights, the variance in the outcomes, and the correlation between the estimated weights and the outcome. Simulation results are presented in both Table \ref{tab:drivers} and Figure \ref{fig:ex_1}.

\begin{table}[ht]
 \caption{\label{tab:drivers} Drivers of design sensitivity}
\centering
\begin{tabular}{cccccccccc}
 \toprule 
 \multicolumn{4}{c}{Simulation Parameters} & \\
$\tau$ & $\beta_y$ & $\beta_\pi$ & $\sigma_y$ & $\var(Y \mid Z = 0)$ & $\var(w \mid Z = 0)$ & $\cor(w, Y \mid Z = 0)$ & $\Tilde{\Lambda}$ & $\Tilde{R}^2$ \\
  \midrule
  \multicolumn{5}{l}{Effect Size} \\ \midrule
0.25 &  1 & 1 & 1 & 1.83 & 1.28 & 0.54 & 1.27 & 0.04 \\ 
  0.50 &  1 & 1 & 1 & 1.82 & 1.30 & 0.54 & 1.59 & 0.13 \\ 
  0.75 &  1 & 1 & 1 & 1.83 & 1.35 & 0.53 & 2.01 & 0.24 \\ 
  1.00 &  1 & 1 & 1 & 1.82 & 1.29 & 0.54 & 2.55 & 0.37 \\ 
  1.25 &  1 & 1 & 1 & 1.83 & 1.26 & 0.54 & 3.27 & 0.49 \\ 
  1.50 &  1 & 1 & 1 & 1.83 & 1.30 & 0.54 & 4.16 & 0.57 \\ 
  1.75 &  1 & 1 & 1 & 1.83 & 1.29 & 0.54 & 5.35 & 0.64 \\ 
  2.00 &  1 & 1 & 1 & 1.83 & 1.35 & 0.53 & 7.02 & 0.70 \\ 
   \midrule
   \multicolumn{5}{l}{Variance in Outcomes} \\ \midrule
   1.00 & 0.5 & 1 & 0.5 &  0.46 & 1.31 & 0.54 & 6.89 & 0.70 \\ 
  1.00 &  1.0 & 1 & 1.0 & 1.83 & 1.30 & 0.54 & 2.53 & 0.37 \\ 
  1.00 & 1.5 & 1 & 1.5 & 4.11 & 1.31 & 0.53 & 1.86 & 0.20 \\ 
  1.00 & 2.0 & 1 & 2.0 & 7.33 & 1.29 & 0.54 & 1.58 & 0.13 \\ 
  1.00 & 2.5 & 1 & 2.5 & 11.42 & 1.29 & 0.54 & 1.45 & 0.09 \\ \midrule
   \multicolumn{5}{l}{Variance in Weights}\\\midrule
   1.00 & 0.80 & 0.50 & 1.11 & 1.84 & 0.27 & 0.54 & 2.55 & 0.74 \\ 
  1.00 & 0.88 & 0.75 & 1.07 & 1.83 & 0.66 & 0.54 & 2.56 & 0.54 \\ 
  1.00 & 1.00 & 1.00 & 1.00 & 1.83 & 1.32 & 0.53 & 2.53 & 0.37 \\ 
  1.00 & 1.17 & 1.25 & 0.88 & 1.82 & 2.30 & 0.54 & 2.54 & 0.25 \\ 
  1.00 & 1.40 & 1.50 & 0.65 & 1.83 & 4.30 & 0.53 & 2.50 & 0.15 \\ 
  1.00 & 1.64 & 1.75 & 0.14 & 1.83 & 6.52 & 0.54 & 2.60 & 0.11 \\ \midrule
   \multicolumn{5}{l}{Correlation between Weights and Outcomes}\\\midrule
   1.00 & -1.4 & 1 & 0.46 & 1.83 & 1.29 & -0.75 & 2.56 & 0.49 \\ 
  1.00 & -1.2 & 1 &  0.80 & 1.83 & 1.31 & -0.64 & 2.56 & 0.41 \\ 
  1.00 & -1.0 & 1 & 1.00 & 1.82 & 1.32 & -0.53 & 2.55 & 0.37 \\ 
  1.00 & -0.8 & 1 & 1.14 & 1.83 & 1.26 & -0.43 & 2.55 & 0.34 \\ 
  1.00 & -0.6 & 1 & 1.24 & 1.83 & 1.27 & -0.32 & 2.55 & 0.32 \\ 
  1.00 & -0.4 & 1 & 1.30 & 1.83 & 1.27 & -0.21 & 2.56 & 0.31 \\ 
  1.00 & -0.2 & 1 & 1.34 & 1.83 & 1.31 & -0.11 & 2.56 & 0.30 \\ 
  1.00 & 0.0 & 1 & 1.35 & 1.83 & 1.29 & 0.00 & 2.56 & 0.30 \\ 
  1.00 & 0.2 & 1 & 1.34 & 1.83 & 1.32 & 0.11 & 2.56 & 0.29 \\ 
  1.00 & 0.4 & 1 & 1.30 & 1.83 & 1.27 & 0.22 & 2.56 & 0.31 \\ 
  1.00 & 0.6 & 1 & 1.24 & 1.83 & 1.30 & 0.32 & 2.55 & 0.32 \\ 
  1.00 & 0.8 & 1 & 1.14 & 1.83 & 1.29 & 0.43 & 2.56 & 0.34 \\ 
  1.00 & 1.0 & 1 & 1.00 & 1.83 & 1.31 & 0.53 & 2.52 & 0.36 \\ 
  1.00 & 1.2 & 1 & 0.80 & 1.82 & 1.28 & 0.64 & 2.57 & 0.42 \\ 
  1.00 & 1.4 & 1 & 0.46 & 1.83 & 1.31 & 0.75 & 2.53 & 0.49 \\ \bottomrule
\end{tabular}
\end{table}

% \FloatBarrier
\subsection{Treatment Effect Heterogeneity} \label{subsec:het} 

To allow for treatment effect heterogeneity, we modify the data generating process under the favorable situation \eqref{eq:sim_dgp} so that it allows for the individual treatment effect for unit $i$ to depend on its covariate value $X_i$:
\begin{align}
    Y_i = \beta_y X_i + \tau_i Z_i + u_i, \ \ \ \ \ \ \tau_i = \tau_0 + \beta_\tau X_i. \label{eq:sim_hetero}
\end{align}
In this setup, $\beta_\tau$ controls the degree of treatment effect heterogeneity. As $\beta_\tau$ increases in magnitude, the individual treatment effects depend more on the covariate values, while $\beta_\tau = 0$ recovers the constant effects case. When $\beta_\tau$ and $\beta_\pi$ are the same sign, the weights and individual treatment effects are positively correlated; otherwise, they are negatively correlated.

As discussed in Section \ref{sec:des_trimming}, the impact of trimming on design sensitivity can depend on the correlation between the weights and the treatment effects. For positive correlation, trimming units with large weights also removes units with larger treatment effects, reducing the ATT and thus the design sensitivity. The reverse is true when the weights and effect sizes are negatively correlated.

The simulations presented in Table \ref{tab:hetero_sims} examine the impact of trimming on design sensitivity for varying levels of treatment effect heterogeneity. In the constant effects case with $\beta_\tau = 0$, trimming improves design sensitivity for both the variance-based and marginal sensitivity models. As $\beta_\tau$ decreases, trimming increases the treatment effect and thus increases the design sensitivity compared to trimming with constant effects. Conversely, the treatment effect and design sensitivity decrease as $\beta_\tau$ increases, eventually causing trimming to hurt design sensitivity compared to not trimming at all. Higher levels of effect heterogeneity are required for trimming to hurt design sensitivity for the variance-based sensitivity model than the marginal sensitivity model.

\begin{table}[!ht]
\centering 
\textbf{Impact of trimming on design sensitivity for varying treatment effect heterogeneity} \\ \vspace{2mm} 
\begin{tabular}{lccccccccc}
\toprule 
$\beta_\tau$ & Trimmed ATT & $\tilde \Lambda$ & $\tilde \Lambda_{trim}$ & Change & $\tilde R^2$ & $\tilde R^2_{trim}$ & Change \\ \midrule 
-1.5 & 2.64 &  6.10 & 15.34 & 9.24 & 0.42 & 0.73 & 0.31 \\
-1.0 & 2.49 &  6.55 & 12.86 & 6.31 & 0.44 & 0.71 & 0.27 \\
-0.5 & 2.36 & 6.38 & 10.80 & 4.42 &0.44 & 0.69 & 0.24 \\
0.0 & 2.23 &  6.22 & 9.19 & 2.97 & 0.43&  0.66 & 0.23 \\
0.5 & 2.10 &  6.55 & 8.01 & 1.46 & 0.44 & 0.63 & 0.19 \\
1.0 & 1.96 &  6.24 & 6.82 & 0.58 & 0.44 & 0.60 & 0.16 \\
1.5 & 1.83 &  6.27 & 5.89 & -0.38 & 0.43 & 0.57 & 0.14 \\
2.0 & 1.70 &  6.06 & 5.08 & -0.98 & 0.43 & 0.53 & 0.10 \\
3.0 & 1.43 &  6.47 & 3.81 & -2.66 & 0.44 & 0.44 & -0.00 \\
4.0 & 1.17 &  6.16 & 2.97 & -3.19 & 0.44 & 0.35 & -0.09 \\ 
\bottomrule 
\end{tabular}
\caption{We vary the amount of treatment effect heterogeneity and assess the impact of trimming on design sensitivity. For the simulation, we set the average treatment effect to be equal to 2.23, $\var(Y \mid Z = 0) = 2.48$, $\var(Y \mid Z = 0, w<m) = 2.4$, $\cor(w, Y \mid Z = 0) = 0.5$, $\cor(w, Y \mid Z = 0, w < m) = 0.6$.} 
\label{tab:hetero_sims}
\end{table}

\subsection{Assessing Augmentation under Model Misspecification}

Theorem \ref{thm:aug} establishes the conditions under which outcome model augmentation improves design sensitivity compared to a standard weighted estimator for the variance-based sensitivity model. We now evaluate the impact of model specification on design sensitivity through simulation. The results are displayed in Table \ref{tab:aug}. In line with the data generating process in the favorable situation \eqref{eq:sim_dgp}, the outcome $Y$ is modelled as a linear function of $X$ for the correctly specified case. We also consider several model misspecifications, with each misspecified model replacing $X$ with $W$. We consider a noise model with $W_i \overset{iid}{\sim} N\left(0, 1\right)$, $W_i = X_i^3$ for misspecification 1, $W_i = \exp\left\{X_i / 2\right\}$ for misspecification 2, and $W_i = \log\left(X_i^4\right)$ for misspecification 3.

The correctly specified model yields the largest improvements in design sensitivity compared to weighting alone for both sensitivity models. The improvement in design sensitivity stems from the reduction in the variance of the outcome from augmentation, which can be seen by comparing the variance of $Y$ to the variance of the residual $e$, where the residual $e$ plays the role of a pseudo outcome for the augmented weighted estimator. The design sensitivities are unchanged for the noise model; however, it is not advisable to implement this model in practice since it could lead to less precise estimates in a finite-sample. While performing augmentation with the first two misspecified models does not help design sensitivity as much as with the correctly specified model, both models result in higher design sensitivities than the standard weighted estimator, highlighting that even misspecified outcome models could lead to improvements. On the other hand, the third misspecified model yields lower design sensitivity values for the marginal sensitivity model and only minor improvements for the variance-based sensitivity model.

%\FloatBarrier
\begin{landscape} 

%\FloatBarrier
\begin{table}[!ht]
\footnotesize
\centering
\textbf{Impact of augmentation on design sensitivity under outcome model misspecification} \\ \vspace{2mm} 
\begin{tabular}{lcccccccccc}
\toprule 
& \multicolumn{2}{c}{Standard IPW} & \multicolumn{2}{c}{Augmented IPW} \\
$\sigma^2_y$ & $\cor(w, Y \mid Z = 0)$ & $\var(Y \mid Z = 0)$ & $\cor(e,Y \mid Z = 0)$ & $\var(e \mid Z = 0)$ & $\tilde \Lambda$ & $\tilde \Lambda_{aug}$ & Change & $\tilde R^2$ & $\tilde R^2_{aug}$ & Change \\ \midrule
\multicolumn{3}{l}{Outcome Model Type: Correct} \\ \midrule 
0.50 & 0.70 & 1.08 & -0.00 & 0.25 & 6.60 & 87.64 & 81.04 & 0.76 & 0.87 & 0.12 \\
1.00 & 0.54 & 1.84 & 0.00 & 1.01 & 4.15 & 7.18 & 3.03 & 0.57 & 0.63 & 0.06 \\
1.50 & 0.42 & 3.08 & 0.00 & 2.25 & 2.97 & 3.59 & 0.62 & 0.41 & 0.44 & 0.03 \\
2.00 & 0.33 & 4.84 & 0.00 & 4.01 & 2.37 & 2.58 & 0.21 & 0.28 & 0.30 & 0.01 \\
3.00 & 0.23 & 9.81 & -0.00 & 8.97 & 1.83 & 1.88 & 0.05 & 0.16  & 0.16 & 0.00 \\
\midrule
\multicolumn{3}{l}{Outcome Model Type: Noise} \\\midrule 
0.50 & 0.70 & 1.08 & 0.70 & 1.08 & 6.50 & 6.50 & 0.00 & 0.76 & 0.76 & 0.00 \\
1.00 & 0.54 & 1.83 & 0.54 & 1.83 & 4.18 & 4.18 & 0.00 & 0.58 & 0.58 & 0.00 \\
1.50 & 0.41 & 3.08 & 0.41 & 3.08 & 2.94 & 2.94 & 0.00 & 0.40 & 0.40 & -0.00 \\
2.00 & 0.33 & 4.84 & 0.33 & 4.84 & 2.36 & 2.36 & 0.00 & 0.28 & 0.28 & -0.00 \\
3.00 & 0.23 & 9.79 & 0.23 & 9.79 & 1.84 & 1.84 & 0.00 & 0.16 & 0.16 & 0.00 \\
\midrule
\multicolumn{3}{l}{Outcome Model Type: Misspecification 1} \\\midrule 
0.50 & 0.69 & 1.08 & 0.46 & 0.61 & 6.46 & 22.90 & 16.44 & 0.75 & 0.78 & 0.03 \\
1.00 & 0.54 & 1.83 & 0.31 & 1.37 & 4.17 & 5.62 & 1.45 & 0.57 & 0.59 & 0.01 \\
1.50 & 0.41 & 3.09 & 0.23 & 2.62 & 2.95 & 3.28 & 0.33 & 0.40 & 0.41 & 0.01 \\
2.00 & 0.33 & 4.82 & 0.17 & 4.35 & 2.36 & 2.49 & 0.13 & 0.28 & 0.29 & 0.01 \\
3.00 & 0.23 & 9.82 & 0.12 & 9.34 & 1.83 & 1.86 & 0.03 & 0.16 & 0.16 & 0.00 \\
\midrule
\multicolumn{3}{l}{Outcome Model Type: Misspecification 2} \\\midrule 
0.50 & 0.69 & 1.08 & -0.15 & 0.34 & 6.43 & 45.31 & 38.88 & 0.75 & 0.84 & 0.09 \\
1.00 & 0.54 & 1.83 & -0.08 & 1.08 & 4.12 & 6.27 & 2.15 & 0.57 & 0.62 & 0.05 \\
1.50 & 0.41 & 3.10 & -0.06 & 2.35 & 2.98 & 3.41 & 0.43 & 0.40 & 0.43 & 0.02 \\
2.00 & 0.33 & 4.84 & -0.04 & 4.09 & 2.36 & 2.53 & 0.17 & 0.29 & 0.30 & 0.01 \\
3.00 & 0.23 & 9.86 & -0.03 & 9.12 & 1.84 & 1.87 & 0.03 & 0.16 & 0.16 & 0.00 \\
\midrule
\multicolumn{3}{l}{Outcome Model Type: Misspecification 3} \\\midrule 
0.50 & 0.70 & 1.08 & 0.74 & 0.96 & 6.62 & 4.67 & -1.95 & 0.76 & 0.80 & 0.04 \\
1.00 & 0.53 & 1.83 & 0.55 & 1.71 & 4.15 & 3.50 & -0.65 & 0.56 & 0.59 & 0.02 \\
1.50 & 0.42 & 3.08 & 0.42 & 2.96 & 2.96 & 2.74 & -0.22 & 0.40 & 0.41 & 0.01 \\
2.00 & 0.33 & 4.83 & 0.34 & 4.71 & 2.37 & 2.28 & -0.09 & 0.28 & 0.29 & 0.01 \\
3.00 & 0.23 & 9.80 & 0.23 & 9.68 & 1.84 & 1.82 & -0.02 & 0.16 & 0.16 & 0.00\\ \bottomrule 
\end{tabular}
\caption{For the correctly specified outcome model, we model the control potential outcome $Y(0)$ as a linear function of $X$. For other outcome models, we replace $X$ with $W$, where $W_i \overset{iid}{\sim} N(0, 1)$ for the noise model, $W_i = X_i^3$ for misspecification 1, $W_i = \exp\left\{X_i / 2\right\}$ for misspecification 2, and $W_i = \log\left(X_i^4\right)$ for misspecification 3.}
\label{tab:aug} 
\end{table}
\end{landscape} 
\endgroup 
\end{document}